\documentclass[11pt]{article}
\usepackage{fullpage}
\usepackage{amssymb,amsmath}
\usepackage{epsfig}
\usepackage{palatino}
\usepackage{graphicx}
\usepackage{textcomp}
\usepackage{amsthm}
\usepackage{color}
\usepackage{tocloft}
\usepackage[colorlinks=true,linkcolor=blue,linktoc=all]{hyperref}
\newtheorem{theorem}{Theorem}
\newtheorem{proposition}[theorem]{Proposition}

\newtheorem{definition}[theorem]{Definition}

\newtheorem{lemma}[theorem]{Lemma}
\newtheorem{conjecture}[theorem]{Conjecture}

\newtheorem{corollary}[theorem]{Corollary}
\newtheorem{fact}[theorem]{Fact}

\newtheorem*{theorem*}{Theorem}
\newtheorem*{lemma*}{Lemma}

\newcommand{\nc}{\newcommand}
\nc{\rnc}{\renewcommand}

\def\ba#1\ea{\begin{align}#1\end{align}}
\def\bas#1\eas{\begin{align*}#1\end{align*}}
\def\bpm#1\epm{\begin{pmatrix}#1\end{pmatrix}}
\nc{\nn}{\nonumber}
\nc{\eq}[1]{(\ref{eq:#1})}
\nc{\eqs}[2]{(\ref{eq:#1}) and (\ref{eq:#2})}
\newcommand{\vev}[1]{\left\langle #1\right\rangle}
\newcommand{\braket}[2]{\left\langle #1\middle|#2\right\rangle}
\def\begsub#1#2\endsub{\begin{subequations}\label{eq:#1}\begin{align}#2\end{align}\end{subequations}}
\nc\qand{\qquad\text{and}\qquad}
\nc\mnb[1]{\medskip\noindent{\bf #1}}

\nc\benum{\begin{enumerate}}
\nc\eenum{\end{enumerate}}
\nc\bit{\begin{itemize}}
\nc\eit{\end{itemize}}
\nc{\ot}{\otimes}
\rnc{\L}{\left} 
\nc{\R}{\right}

\newcommand{\secref}[1]{Section~\ref{sec:#1}}

\newcommand{\thmref}[1]{Theorem~\ref{thm:#1}}
\newcommand{\propref}[1]{Proposition~\ref{prop:#1}}

\newcommand{\defref}[1]{Definition~\ref{def:#1}}
\newcommand{\corref}[1]{Corollary~\ref{cor:#1}}

\newcommand{\factref}[1]{Fact~\ref{fact:#1}}

\def\bbC{\mathbb{C}}

\DeclareMathOperator*{\bbE}{\mathbb{E}}
\def\bbF{\mathbb{F}}

\def\calB{{\cal B}}
\def\calC{{\cal C}}
\def\calD{{\cal D}}

\def\calU{{\cal U}}

\def\R{\mathbb{R}}

\def\F{\mathbb{F}}
\def\mod{\mbox{mod}}

\def\log{{\rm log}}
\def\supp{{\rm supp}}
\def\Prob{{\rm Prob}}

\newcommand{\be}{\begin{eqnarray}}
\newcommand{\ee}{\end{eqnarray}}

\newcommand\ket[1]{{ |{#1} \rangle }}
\newcommand\bra[1]{{ \langle {#1} | }}
\newcommand{\proj}[1]{\left|#1\right\rangle\left\langle #1\right|}
\newcommand\ketbra[1]{{\ket{#1}\bra{#1}}}

\def\QMA{{\sf{QMA}}}
\def\NP{{\sf{NP}}}

\def\NLTS{{\sf{NLTS}}}
\def\NLETS{{\sf{NLETS}}}
\def\cNLTS{{\sf{cNLTS}}}
\def\PCP{{\sf{PCP}}}
\def\qPCP{{\sf{qPCP}}}
\def\qLTC{{\sf{qLTC}}}
\def\qLDPC{{\sf{qLDPC}}}
\def\LTC{{\sf{LTC}}}
\def\QNC{{\sf{QNC}}}

\DeclareMathOperator{\diag}{diag}
\DeclareMathOperator{\dist}{dist}
\DeclareMathOperator{\Img}{Im}
\DeclareMathOperator{\Span}{Span}
\DeclareMathOperator{\tr}{tr}

\newcommand{\ignore}[1]{}

\newcommand{\eps}{\varepsilon}
\renewcommand{\epsilon}{\varepsilon}

\nc{\hin}{h_{\text{in}}}
\nc{\pin}{\partial_{\text{in}}}
\nc{\pell}{\partial_{\ell}}

\newcommand{\nocontentsline}[3]{}
\newcommand{\tocless}[2]{\bgroup\let\addcontentsline=\nocontentsline#1{#2}\egroup}
\makeatletter
\newcommand{\cftsectionprecistoc}[1]{\addtocontents{toc}{%
  {\leftskip \cftsecindent\relax
   \advance\leftskip \cftsecnumwidth\relax
   \rightskip \@tocrmarg\relax
   \textit{#1}\protect\par}}}
\makeatother

\title{Local Hamiltonians Whose Ground States are Hard to Approximate}

\author{Lior Eldar\thanks{Center for Theoretical Physics, MIT} \and
  Aram W.~Harrow\footnotemark[1]}

\begin{document}

\maketitle

\abstract{ 

  Ground states of local Hamiltonians can be generally highly entangled: any quantum
  circuit that generates them, even approximately, must be sufficiently deep to allow
  coupling (entanglement) between any pair of qubits.  Until now this property was not
  known to be "robust" - the marginals of such states to a subset of the qubits containing
  all but a small constant fraction of them may be only locally entangled, and hence
  approximable by shallow quantum circuits.  In this work we construct a family of 16-local
  Hamiltonians for which any marginal of a ground state to a fraction at least $1-10^{-9}$
  of the qubits must be globally entangled. 
  
  This provides evidence that quantum entanglement is not very
  fragile, and perhaps our intuition about its instability is an
  artifact of considering local Hamiltonians which are not only local
  but {\it spatially local}.  Formally, it provides positive evidence
  for two wide-open conjectures in condensed-matter physics and
  quantum complexity theory which are the $\qLDPC$ conjecture,
  positing the existence of "good" quantum LDPC codes, and the $\NLTS$
  conjecture \cite{FreedmanH13} positing the existence of local
  Hamiltonians in which any low-energy state is highly entangled.
 
  Our Hamiltonian is based on applying the hypergraph product by
  Tillich-Z\'emor \cite{TZ09} to a classical locally testable code.  A
  key tool in our proof is a new lower bound on the vertex expansion
  of the output of low-depth quantum circuits, which may be of
  independent interest.
  
}

\tocless\section{Introduction}

\subsection{Background and main result}
\subsubsection{Multiparticle entanglement, trivial states and topological
  order}
Quantum mechanics has overturned our classical intuitions about the
nature of information, computing and knowledge.  Perhaps the greatest
departure from earlier notions of information is the phenomenon of
entanglement in which a many-body quantum state cannot be reduced to a probabilistic mixture of descriptions
of the state of each individual particle.  For decades, entanglement
was viewed in terms of its counterintuitive properties, e.g.~ the Bell
and GHZ ``paradoxes,'' and only in recent years has quantum
information theory begun a systematic program of quantifying,
characterizing and finding ways to test entanglement.

However, in typical many-body systems, 
and from a complexity-theoretic point of view,
the
important question is not to establish the existence of entanglement, but rather to determine the {\it complexity} of the quantum circuit required to generate it.  Many of the results of quantum information
theory apply to the case of bipartite entanglement and often do not
extend to this setting of large numbers of interacting systems.
 For example, a collection of $n/2$
singlets has high bipartite entanglement across most cuts but this
entanglement is in a certain sense ``local'' and could be eliminated
by a suitable coarse-graining.  

The concept of a ``trivial state'' is meant to express the notion that
states such as $n/2$ singlets have only {\it low-complexity} entanglement,
and relates to the circuit complexity of generating quantum states:
\begin{definition}[Depth-$d$ Trivial States]
We say that an $n$-qubit state $\rho$ is depth-$d$ trivial if it can be prepared
by applying a depth-$d$ quantum circuit comprised of $d$ layers of tensor-products of $2$-local quantum gates, to $\ket{0}^{\ot N}$ (for some
$N\geq n$) and tracing out $N-n$ qubits. 
\end{definition}
This is a special case of a
more general classification of quantum phases of matter in which two states are
said to be equivalent if they differ by an $O(1)$-depth quantum
circuit~\cite{CGW10}; here trivial states correspond to the phase that includes
product states.
Nontrivial states are
sometimes said to be topologically ordered, and examples include code
states of the toric code, or indeed any QECC with distance more than a
constant~\cite{BHF06}.  ``Topological order'' is an imprecisely
defined term that we will not do justice to here, but trivial states
have been shown to be equivalent to states without [various versions
of] topological order in \cite{BHF06, Haah14, kitaev13, SM14,Kim14}.

\subsubsection{The Physical Perspective: Robustness of Entanglement}

Arguably, the biggest barrier to building a quantum computer is 
quantum decoherence, which is the process by which
long-range entanglement, i.e.~the type that could be useful
to solve hard computational problems, e.g.~in Shor's algorithm, evolves into
classical distributions of trivial states, by interacting with the
environment.  With enough decoherence, classical computers can simulate
the quantum one, thereby extinguishing all hope for a quantum speed-up.

To counter these environmental errors, one must then use quantum codes,
which spread-out the quantum information over a larger space,
and so the introduced redundancy then adds some resilience to the computation we are trying to perform.
Indeed, the fault-tolerance theorem (see \cite{NC11}) uses quantum codes to argue that universal quantum computation
can be carried out efficiently under uniformly random error of sufficiently small constant rate.

However, in some cases the uniform random error model may be insufficient, and we would want to consider an error model that is much more adversarial than random.
In particular, one might ask a much simpler question: 
how much entanglement is left in the system if we trace-out,
or damage in some way, a small, yet constant fraction of the qubits.
In general, when physicists have considered locally-defined quantum-mechanical systems,
the immediate notion was to consider regular grids of $2$ or $3$ dimensions.
It can be easily shown that the quantum codes considered on such lattices
easily lose all long-range entanglement by acting on some small constant fraction of all qubits.
This gave rise to the folklore notion that quantum systems cannot posses
a 
robust form of quantum entanglement: namely the property that even quantum states that "pass as groundstates"
on most qubits are non-trivial (highly entangled).

Hence, the above problem raises a fundamental question regarding
quantum entanglement: could it be that our notion that entanglement is
fragile is merely an artifact of building systems in low-dimensional
grid?  Could it be that, at least theoretically, quantum systems
defined on more highly connected topologies could have entanglement
which is more resilient?  
A conjecture of this form was formulated rigorously by
Freedman and Hastings \cite{FreedmanH13} and was called the $\NLTS$
conjecture (see definition \ref{def:nlts}).  
\begin{definition}

\textbf{No Low-Energy Trivial States (NLTS)} 

\noindent
Let $\{H_n\}_{n\in \mathbf{N}}$ be a family of $k$-local Hamiltonians
for $k=O(1)$. 
We say that $\{H_n\}_{n\in \mathbf{N}}$ is $\eps$-NLTS if 
there exists a constant $\eps>0$ such that
for any $d$ and all sufficiently large $n$, 
the following holds: 
if $\{\rho_n\}$ is any family of $d$-trivial states then
\be \tr[\rho_n H_n] >\lambda_{\min}(H_n) + \eps.
\ee
\end{definition}


Several works around the $\NLTS$ conjecture, and its parent conjecture
(quantum $\PCP$ - see next section) have provided ambiguous 
evidence
about its ultimate status.  Indeed, the works of
\cite{BV05,BH-product,Has12,AE14} have suggested that the $\NLTS$
conjecture may be false by showing that large classes of local
Hamiltonians have trivial states at very low energies.  
Moreover, some of these results study local Hamiltonians on topologies which
are highly expanding, and correspond, in a sense, to classical 
problems which are hard-to-approximate, because they resist the trivial divide-and-conquer
strategy above.

In this work, we provide a positive indication towards the $\NLTS$ conjecture, and hence the 
$\qPCP$ conjecture by resolving a weaker version of $\NLTS$
that considers errors instead of violations:

\begin{definition}

\textbf{Ground-state impostors}

\noindent
Let $H$ be a Hamiltonian.
A quantum state $\rho$ is said to be an $\eps$-impostor for $H$,
if there exists a set $S\subseteq [n], |S| \geq (1-\eps) n$ 
and a ground state $\sigma$ (i.e.~satisfying
$\tr[H\sigma]=\lambda_{\min}(H)$)
such that
$\rho_S = \sigma_S$.
\end{definition}

\begin{definition}\label{def:nlets}

\textbf{No Low-Error Trivial States (NLETS)} 

\noindent
Let $\{H_n\}_{n\in \mathbf{N}}$ be a family of $k$-local Hamiltonians for $k=O(1)$.
We say that $\{H_n\}_{n\in \mathbf{N}}$ is NLETS if there exists a constant $\eps>0$
such that the following holds:
for any $d$ 
and all sufficiently large $n$, 
if ${\cal F} = \{\rho_n\}$ is any family of $\eps$-impostor states for $\{H_n\}_n$
then
${\cal F}$ is not $d$-trivial.

\end{definition}

By definition any family of bounded-degree local Hamiltonians that is $\NLTS$ is also $\NLETS$:
if a quantum state agrees with a ground state of the quantum system on ``most'' qubits,
then the bounded-degree assumption means that such a state also has low-energy w.r.t.~the
Hamiltonian.  We discuss the difference between $\NLTS$ and $\NLETS$
further in \secref{robust-zoo}.
Our main theorem is as follows:
\begin{theorem}\label{thm:nlets}

\textbf{Explicit NLETS}

\noindent
There exists constants $\eps = 10^{-9}, a,b>0$ and an explicit infinite family of Hamiltonians
$\{H_n\}_n$, each of the form
\be H_n = \frac{1}{m}\sum_{i=1}^m \frac{I + P_i}{2},\ee
for $P_i$ equal to $\pm 1$ times a tensor product of Pauli matrices on $16$ qubits and
identity elsewhere.  
These Hamiltonians have the property that
\bit
\item There exists a state $\ket{\phi_n}$ such that $H_n\ket{\phi_n}=0$.
\item For any  $\eps$-impostor $\rho_n$ for $H_n$
and any  quantum circuit $U_n$ of depth at most $d=b\cdot\log(n)$, we have
\be 
\|\rho_n - U_n \proj{0^{\otimes n}}U_n^\dag \|_1 > n^{-a}. 
\ee
\eit
\end{theorem}

Despite being possibly weaker than $\NLTS$ in terms of the approximation criterion,
our theorem is stronger than $\NLETS$ in the following ways:
\begin{enumerate}
\item The original definition of $\NLTS$ includes a restriction that the generating
  circuit $U_n$ is allowed to couple qubits only if they are coupled via some local term
  of the Hamiltonian $H_n$.  Here we remove this restriction, and show a lower-bound for
  circuits $U_n$ even if they are allowed to couple {\it arbitrary} pairs of qubits.
\item We show that low-depth circuits not only are unable to produce ground states of
  residual Hamiltonians, they cannot even produce states that approximately match the
  classical probability distributions resulting from measuring these states in the $X$ and
  $Z$ bases.
\item We prove a depth lower bound that is not merely $\omega(1)$ but is
  $\Omega(\log(n))$.  A circuit of depth $\Omega(\log(n))$ can potentially generate a
  non-zero correlation between every pair of qubits, hence it ``saturates'' all light-cone
  type arguments.  Since the naive algorithm for estimating expectation values runs in
  time doubly exponential in $d$, our results imply that this algorithm will require time
  $2^{n^{\Omega(1)}}$.
\item We not only show that $d$-trivial states cannot be
  $\eps$-impostors, but we show that these sets of quantum states are separated by a
  trace distance of $n^{-\Omega(1)}$.
\item We use a relatively simple form of Hamiltonian, consisting only of commuting 16-local
  Pauli terms. 
  \end{enumerate}

$\NLETS$.
While previously known constructions of local Hamiltonians are not $\NLETS$, this may in
part be because they are either embedded on a regular grid in low dimensions, or depart
from this in ways that allow for efficient classical description.  Thus, our theorem
suggests that the apparent fragility of many-body entanglement from these examples may be
simply a sign of not considering a wide enough range of examples.

\subsubsection{Robust Entanglement Zoo}\label{sec:robust-zoo}

Given the numerous open problems / conjectures mentioned in this paper, it may be useful
to consider their interaction via a "zoo of robust entanglement" (see Figure \ref{fig:zoo}).
We first list for self-inclusiveness the relevant problems and their definitions.  In what
follows $\eps$ is a positive constant that can be arbitrarily small.
\begin{enumerate}
\item
$\NLTS$ - There exist local Hamiltonians such that any low-energy state is non-trivial.
\item 
$\cNLTS$ - There exist local Hamiltonians such that any quantum state satisfying a
  $\geq 1-\eps$ fraction of all local terms is non-trivial.
\item
$\NLETS$ - There exist local Hamiltonians such that any quantum state that is equal to a
ground state
up to a unitary incident on at most an $\eps$ fraction of qubits, is non-trivial.
\item
$\qLTC$ - There exist local Hamiltonians for which the energy of a quantum state is proportional to its
distance from the ground-space of the Hamiltonian.
\item
$\qLDPC$ - There exist quantum codes with local checks, and minimal distance scaling linearly in the number of qubits.
\item
$\qPCP$ - It is as hard to approximate the ground energy of a local Hamiltonian to a constant fraction accuracy, as it is
to estimate it to inverse-polynomial accuracy.
\end{enumerate}

In Theorem \ref{thm:qltc} we show that an affirmative resolution of the $\qLTC$ conjecture would imply $\NLTS$.
This connects two important conjectures in quantum Hamiltonian complexity, and as described in the next
section, allows us to connect the $\NLTS$ conjecture, to open problems arising in algebraic topology
in the context of high-dimensional expanders.

In our main \thmref{nlets}, we will use a {\it residual} form of quantum local testability to show that our local Hamiltonian 
is $\NLETS$.  Essentially, the more restrictive error model of $\NLETS$ will allow us to leverage
a weaker form of quantum local testability to argue a circuit lower-bound on ground-state impostors.

To stress the difference between $\NLTS$ and $\NLETS$: States that are low energy
w.r.t.~some Hamiltonian may not necessarily corresponds to applying a small
constant-weight error to some ground state of the system, or even to a superposition of
such states.  To make this logical step, one needs to argue that the Hamiltonian has some
form of local testability, with $\qLTC$ being the strongest version thereof.  In other
words: while for any bounded-degree local Hamiltonian, small weight errors translate to
small-weight violations, the converse only holds if the Hamiltonian (or code) is somewhat
$\qLTC$. 

On the other hand, $\NLETS$ is sufficiently general (or "weak") to serve also as a necessary condition for the qLDPC (and $\qLTC$) conjecture:
if there exists a 
quantum code with linear distance $\delta_{min}$, then in particular
any error of fractional weight, say $\delta_{min}/2$, cannot make the state trivial.
The $\qLDPC$ conjecture is still wide open despite some recent progress \cite{BH14} so our Theorem \ref{thm:nlets}
can be considered as progress towards $\qLDPC$ from a slightly different angle - that of robust
entanglement as a weak form of proper quantum error correction.
Hence $\NLETS$ is a step forward in two hierarchies: one is the hardness-of-approximation chain $\qPCP \Rightarrow \NLTS \Rightarrow \cNLTS \Rightarrow \NLETS$,
and the other is the robust-coding chain  $\qLTC \Rightarrow \qLDPC \Rightarrow \NLETS$.

\begin{figure}
\center{
 \epsfxsize=2in
 \epsfbox{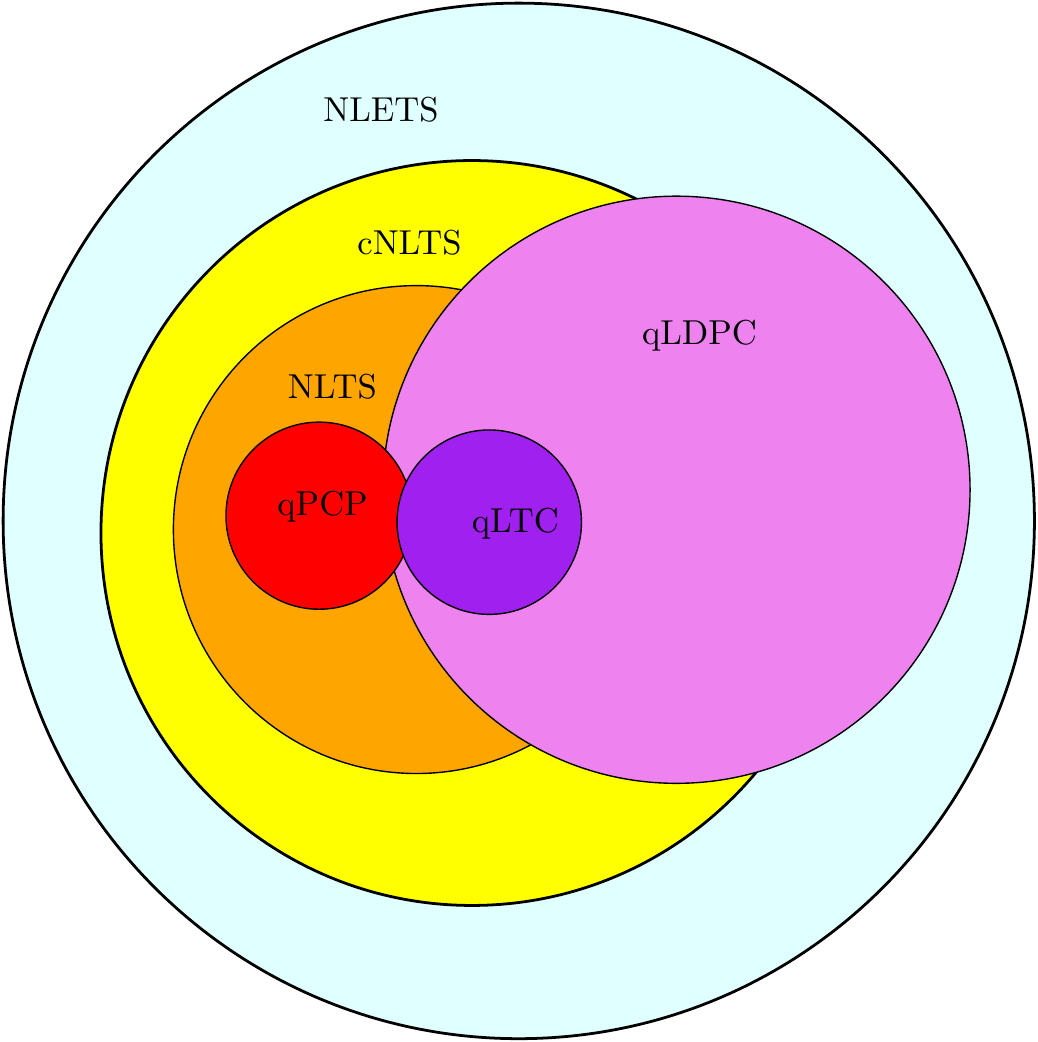}
 \caption{\footnotesize{The Robust-Entanglement Zoo}\label{fig:zoo}}}
\end{figure}

\subsection{The Topological Perspective: High-Dimensional Expanders}

The quest for robust forms of quantum entanglement (e.g.~the $\qPCP$ conjecture) has
raised intriguing questions about which interaction {\it topologies} may be suitable
for such a phenomenon.

As mentioned above, at the very least one would like a topology that is ``expanding'',
i.e.~one in which discarding a constant fraction of the terms would not break the local
Hamiltonian in question into small disjoint components.  However, previous works
\cite{Has12,BH-product,AE14} have indicated that mere graph-expansion may be insufficient, and
hence a more refined, high-dimensional property may be required.

This implies a connection to the nascent field of high-dimensional expanders:
In an attempt to repeat the enormous success of expander graphs researchers have recently
tried to provide a ``standard'' definition of a high-dimensional expander that would allow
simultaneous characterization of these objects from both the combinatorial and spectral perspectives,
as in expander graphs \cite{LinialMeshulam,MeshulamWallach,Gromov} .

In this work, we consider a definition of high-dimensional expansion due to \cite{KKL14,EK16} on complex chains over $\F_2$.  
It is called co-systole expansion:
Let $C$ be a $d$-dimensional complex chain over $\F_2$, $C = \{C_0,\hdots, C_{d-1}\}$ with
boundary maps $\delta_k: C_k \mapsto C_{k-1}$ for all $k\in [d]$.  
Each $C_i$ is a vector space over $\F_2$, and the (linear) boundary maps
have the following defining property:
\be
\delta_{k-1} \circ \delta_k = 0, \forall k.
\label{eq:boundary}\ee
Similarly, one can define the $k$-th co-chain $C^k$ as the space of functions $C_k \mapsto \F_2$,
and corresponding maps, called {\it co-boundary} maps, that map $\delta^k : C^k \mapsto C^{k+1}$, and likewise
\be 
\delta^{k+1} \circ \delta^k = 0, \forall k
\label{eq:co-boundary}\ee
Define $Z^k = \ker(\delta^k)$, $B^k = \text{im}(\delta^{k-1})$ for each $k$.
By \eq{co-boundary}, we have $B^k \subseteq Z^k$.
The complex $C$ is said to be an $(\eps,\mu)$-co-systole expander if for all $0\leq k <d$:
\be \min_{a\in Z^k - B^k} |a| \geq \mu n_k \ee
and
\be
\min_{A\in C^k - Z^k}
\frac{\left| \delta(A) \right|}{\min_{z\in Z^k} |A + z|}
\geq \eps n_k,\ee
where $|x|$ is the Hamming weight of $x$, and $n_k = \dim(C^k)$.
An infinite family of complexes $C^{(n)}$ is then said to be a co-systole expander if there exist
$\eps,\mu$ independent of $n$ such that each $C^{(n)}$ is $(\eps,\mu)$ co-systole.
We usually further require that such families have bounded degree, i.e.~where the boundary
operator has only $O(1)$ non-zero entries in each row/column.
We similarly define a systole expander, if the above properties hold also for the boundary maps.
A family is then called a systole/co-systole expander if the property holds simultaneously for both boundary and co-boundary maps.

In the last two years there has been significant progress towards achieving
bounded-degree co-systole expanders. Kazhdan, Kaufman and Lubotzky \cite{KKL14}
constructed an infinite family of $2$-dimensional co-systole expanders,
and subsequent work of Evra and Kaufman \cite{EK16}, provided a construction an explicit $d$-dimensional
co-systole expander of every dimension $d$, both of bounded degree.

This definition of co-boundary expansion has a natural interpretation in the context of quantum
codes (see e.g. \cite{Terhal15}).  It is known that for every $3$-dimensional complex
chain $\{C_0,C_1,C_2\}$, one can associate a
quantum code by canonically translating the boundary maps $\delta_2,\delta_1$ to
tensor-product Pauli operators.  
Here $C_1$ corresponds to the set of qubits, $C_2$ correspond to the $Z$ stabilizers and
$C_0$ to the $X$ stabilizers. In particular, the stabilizer group is defined to be
\be \{ Z^{\delta_2z} : z \in C_2\} \cup \{X^{(\delta_1)^Tx} : x\in C_0\}.
\label{eq:homology-stabilizer}\ee
(See \secref{def-code} for background and definitions of stabilizer codes.)
These operators in \eq{homology-stabilizer} commute because of \eq{boundary} and they
correspond to a quantum error-correcting code whose logical operators are isomorphic to
the quotient $Z^1/B^1$.  The degree of the complex translates directly to the weight of
the check operators of the quantum code.  Finally, if $\{C_0,C_1,C_2\}$ is both a
co-systole and a systole expander then this implies that the code is in fact a quantum
locally testable code $(\qLTC)$ of linear minimal distance (defined formally in
\defref{qltc} below).

Our first main result \thmref{qltc} shows that the Pauli check terms associated to any such $\qLTC$
comprise an $\NLTS$ local Hamiltonian - namely, one in which any low-energy state
can only be approximated by large-depth quantum circuits.
Hence, using the $\qLTC$ formalism, one can derive the following corollary which connects
the existence of systole / co-systole expanders to the $\NLTS$ conjecture:
\begin{corollary}\label{cor:highdim1}
Let $C^{(n)}$ be a $3$-dimensional complex that is an $(\eps,\mu)$ systole/co-systole expander.
Then the Pauli terms associated with $C^{(n)}$ via the CSS formalism, constitute an $\NLTS$.
\end{corollary}
Alternatively, if $\NLTS$ turned out to be false, it would imply a strong non-duality in the following sense:
any $3$-complex that is a co-systole expander, cannot be like wise a systole expander and vice versa.

It is interesting to point out that the recent construction by Evra and Kaufman \cite{EK16}
achieves only a one-sided expansion, namely of the co-boundary map, but
is actually known {\it not} to possess this property for the boundary map.
This is because the constructed complex has boundary / co-boundary maps which behave very differently.
Notably, such an equivalence exists for the Toric Code, but such a code is very
far from being a co-systole expander.
We further note that a systole / co-systole expander of dimension $2$ (namely a complex of triangles) is actually known not to exist by considering these expanders as a system of $3$-local commuting Hamiltonians \cite{AE11}.

A similar property of having boundaries and co-boundaries behave very differently was at
the core of the behavior of the ``one-sided'' $\NLTS$ construction due to Freedman and
Hastings \cite{FreedmanH13}.  In fact, one can check that the Evra-Kaufman construction
\cite{EK16} also implies a one-sided $\NLTS$ by applying Theorem \ref{thm:qltc} in this
paper to one side of the checks, say the Pauli $X$ operators, in the same way as in
Corollary \ref{cor:highdim1} above.  Without including the precise definitions (found in
\cite{EK16}) one can then connect one-sided $\NLTS$ with certain classes of
high-dimensional expanders called ``Ramanujan complexes'' as follows:
\begin{corollary}\cite{EK16}
For every $d$ there exists an infinite family of Ramanujan complexes $\{C^{(n)}\}$ such that the Pauli operators corresponding to its boundary maps
are $d$-local Hamiltonians that are one-sided $\NLTS$.
\end{corollary}

Hence, \thmref{qltc} of this paper presents a connection between two difficult problems, one in quantum complexity and another in algebraic topology:  a $\qPCP$-optimistic view calls for an attempt to construct such expanders, and resolve the $\NLTS$ conjecture in the affirmative, whereas a $\qPCP$-pessimistic approach, would possibly rule out $\NLTS$ using quantum arguments, and by that
provide a negative result to the systole/ co-systole conjecture.

\subsection{Proof Outline}

\subsubsection{Circuit Lower-bounds} 
We begin by defining a ``complexity witness'',
i.e.~a simple-to-verify property that can prove a state is
nontrivial.  
The goal is then to show a local Hamiltonian system for which
such a complexity witness can be found, not only for its ground state, but for any
quantum state which is an $\eps$-impostor.

Our complexity witness is chosen as the following geometric property of quantum states:
we show that
measuring a trivial state in any product basis results in a
probability distribution over $\bbF_2^n$ with high expansion.  It is
well known that the uniform measure on $\bbF_2^n$, or indeed any
product measure, has good expansion properties.  It is not hard to see
this is also true for the output of low-depth classical circuits.  We
extend this to quantum circuits, by using Chebyshev polynomials in a
way inspired by \cite{FT00,AradKLV12area}.   
\begin{theorem}[informal version of \thmref{vertex}]\label{thm:vertex-informal}
Let $N\geq n>0$ be some integers, and $\ket\psi = U\ket{0^N}$ for $U$ a circuit of depth $d$.  Let $p$ be
the probability distribution that results from measuring the first $n$
qubits in the computational basis; i.e.~
\be p(x) = \sum_{y\in \{0,1\}^{N-n}} |\braket{x,y}{\psi}|^2. \ee
Then for any $\ell \geq \alpha\sqrt{n}2^{1.5d}\geq 1$ with $\alpha \leq 1$, we have:
\be h_{\ell}(p) \geq \Omega(\alpha^2)
\ee
\end{theorem}
where $h_{\ell}(p)$ is the vertex expansion of $\bbF_2^n$ endowed with
measure $p$ and edges between all $x,y$ with $\dist(x,y)\leq \ell$.
Vertex expansion is defined precisely in  \secref{expansion}  but roughly speaking
measures the fraction of weight of any subset that should be near its boundary.

We refer to non-expanding distributions as ``approximately
partitioned''; meaning that we can identify two well-separated subsets
$S_0,S_1$ each with large probability measure.  A prototypical example of a state
giving rise to an approximately partitioned distribution is the so-called ``cat-state'' $(\ket{0^n}+\ket{1^n})/\sqrt{2}$.  However, the cat state is not the
unique ground state of any local Hamiltonian~\cite{CJZZ12},
so it is not a good candidate for an $\NLETS$ system.
 Another possibility,
the
uniform distribution of any (classical) code with good distance, is also
approximately partitioned, but simply using the check operators of a
classical code is insufficient since 
any product string state corresponding to a code-word would pass this
test, but is obviously a trivial state.  

An example of a state which is both approximately partitioned {\em and} locally checkable
is a state of a quantum error correcting code
(QECC) with low-weight generators.  QECCs protect quantum information by encoding a given
Hilbert space into a larger Hilbert space in a non-local fashion, so they are natural candidates for
creating robust forms of entanglement.  We will show in
\secref{warmup} the ``warm-up'' result that Hamiltonians corresponding
to a special subclass of QECCs (namely CSS codes) have no {\em zero-}energy trivial states. (Some version of
this claim was folklore~\cite{BHF06,Haah14}.)  Here if we want to
consider local Hamiltonians then it is necessary to restrict to codes
with low-weight check operators, also known as LPDC (low-density
parity-check) codes.

\subsubsection{Local Testability}
Later we will construct a quantum code in Section \ref{sec:nlets} that is robust, meaning
that even states that violate a small constant fraction of constraints would have a
``complexity witness'' to the fact that they are hard to generate using quantum circuits.
A key ingredient is local testability, a property of significant interest in theoretical
CS~\cite{Goldreich}.   
\begin{definition}
\textbf{Classical locally testable code}

\noindent
A code $C\subseteq \bbF_2^n$ is said to be locally testable with
parameters $q,\rho$, if there exists a set of $q$-local check terms
$\{C_1,\ldots,C_m\}$, such that
$$
\Prob_i \left[C_i(w) = 1\right] \geq \rho \cdot \frac{\dist(w, C)}{n}.
$$
\end{definition}
This means that strings violating only a few checks must be close to the code.

In \cite{AE15} the authors define a quantum analog of $\LTC$'s, called $\qLTC$ (see
\defref{qltc} below for a precise definition).  It is an
open question (the ``$\qLTC$ conjecture'') whether there exist $\qLTC$s with constant soundness $\rho>0$
and $q = O(1)$.
In
\secref{qltc-nlts} we show that the $\qLTC$ conjecture implies the $\NLTS$ conjecture.
This is because low-energy quantum states are close to the encoded quantum space, which
has provable circuit lower-bounds.  However, our approach to $\NLTS$ will be instead will
use the simplest classical $\LTC$s, namely the repetition code, and use it to construct a
quantum code with a residual form of local testability.  

\subsubsection{The Hypergraph Product}
The Tillich-Z\'emor hypergraph product~\cite{TZ09} provides a method for embedding a
classical code into a CSS code.  It does not preserve the minimal distance of the
composing codes. Indeed its best distance parameter scales like $O(\sqrt{N})$ for $N$
qubits, whereas we are interested in protecting against syndromes of linear weight.
However, one lesson from our work is that distance may not always be the best measure for
how well a quantum code resists error in the context of $\NLTS$.  As it turns out, there
exists a subset of the logical words of the quantum product code that is isomorphic to the
original $\LTC$.  While these words have weight only $O(\sqrt{N})$ each, they also inherit
some form of local testability from the underlying classical $\LTC$s.

Next, we note that even a residual form of local testability does not imply, on its own,
that a low-energy quantum state is hard to generate.  In particular, as mentioned before,
classical bit-string assignments are trivially easy to generate.  Here, we make use of the
fact the quantum uncertainty principle.  The fact that we use a quantum code means we can
measure in either the $X$ or $Z$ basis.  A standard uncertainty argument means that at
least one of these should have high uncertainty.  
This forces the distribution of any low-violation quantum state not only to be clustered
around the original code-space as with classical $\LTC$s, but also to have considerable
measure on at least two far-away subsets of the code, by the uncertainty principle.
This places a lower bound on how dispersed the distribution
is.  In this context, the distance and local testability of the original $\LTC$ will imply
that at least in one of the $X$ or $Z$ basis, the measured distribution must be
approximately partitioned, and therefore that the state must be nontrivial.

\begin{figure}
\center{
 \epsfxsize=4in
 \epsfbox{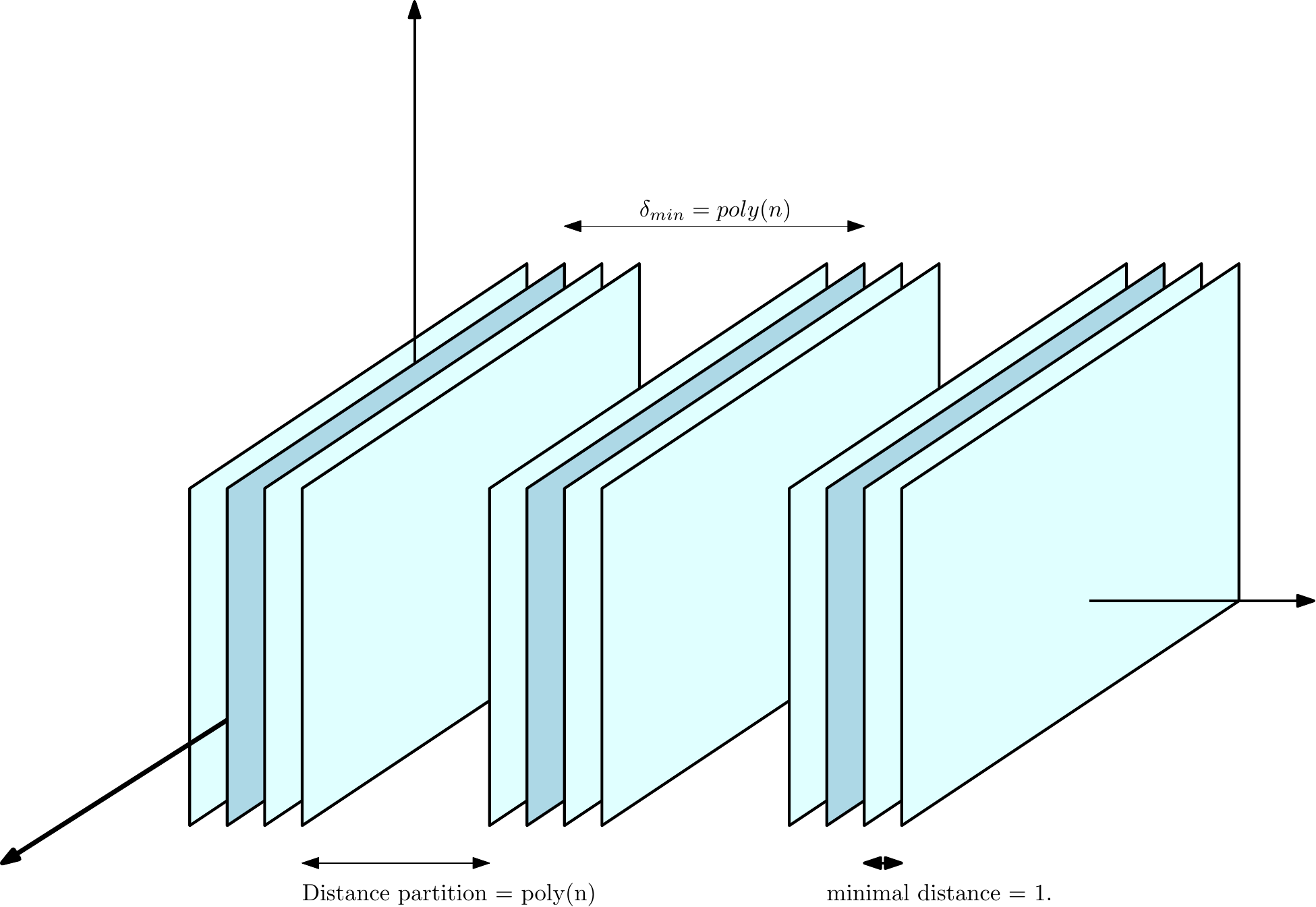}
 \caption{\footnotesize{Depiction of the robustness of qLTC.  The code space should be
     divided into linear subspaces that are separated by distance $\Delta_{\min}$, and
     when noise is added it should still cause the state to be clustered near these
     subspaces.  Thus even though distance drops to 1 when we allow low-energy states,
     there is still an approximate partition of the distribution when we measure in the
     $X$ or $Z$ basis.  Our $\NLETS$ system will rely on a similar partition to achieve
     its circuit lower bound.}}}
\end{figure}

\subsubsection{The Construction}

We provide full details of the construction in Section \ref{sec:nlets}
and here we provide a high-level sketch.  Consider a classical
code with a parity-check matrix $H$, defined with rows corresponding to a set
of $m$ checks
and columns corresponding to $n$ bits, thus $\ker H
\subseteq \bbF_2^n$.
The hypergraph
product of $H$ with itself is a quantum code on $m^2+n^2$ qubits with check matrices $H_x,H_z$
(corresponding to Pauli $X$, Pauli $Z$ operators) defined as follows:
\be H_x = 
\left( H \ot I_n
\middle|
I_m \ot H^T \right)
\qand
H_z = 
\left( I_n \ot H
\middle|
H^T \ot I_m\right)
\ee
To see that the usual orthogonality condition for CSS codes is
satisfied, observe that
\be H_x H_z^T = 
\left( H \ot I_n
\middle|
I_m \ot H^T \right)
\cdot 
\bpm I_n \ot H^T \\ H \ot I_m \epm
= H\ot H^T + H\ot H^T = 0.
\ee

We then start by choosing $H$ as a check matrix of the repetition
code, corresponding to the edges of a $d$-regular expander graph (e.g.~from \cite{LubotzkyPS88}) as equality constraints modulo $\F_2$.
One can check that the locality of the matrices $H_x,H_z$ -
i.e.~the number of $1$'s in each row - 
is the sum of the right and left degree of the Tanner graph defined by $H$.
In the case of the expander graph, this results in degree $d+2$.

\subsection{Previous Work}\label{sec:previous}
Most of the previous results established various settings in which
$\NLTS$ was known not to hold.  This is true even though in all of the following cases (except the first) , the
corresponding classical instance is NP-hard to approximate.    In the following list we identify
Hamiltonians with the [hyper-]graph in which each qubit is a vertex
and each interaction term defines a [hyper-]edge.

\benum
\item Non-expanding graphs, including $O(1)$-dimensional lattices.
  This was folklore and holds even classically, but was formalized in \cite{BH-product}.
\item Graphs where most vertices have $\omega(1)$
  degree~\cite{BH-product}.
\item 2-local Hamiltonians with commuting terms~\cite{BV05}.
\item 3-qubit Hamiltonians with commuting terms~\cite{AE11}.
\item Sparse commuting $O(1)$-local Hamiltonians corresponding to
  graphs with high girth~\cite{Has12}.
\item Commuting $O(1)$-local Hamiltonians with high local expansion~\cite{AE14}.
\eenum

The above results rule out topologies (such as high-degree graphs)
which are known to be central to the classical PCP constructions.
This combination of constraints (degree and expansion must be high
but not too high, etc.) made it plausible that $\NLETS$ would be false, at
least in the commuting case.    Indeed Ref.~\cite{BH-product} shows
that if the gap-amplification step of Dinur's PCP theorem \cite{Din07}
had a quantum analogue with similar behavior (as was proposed by \cite{AALV09}), it would actually
{\em refute} the qPCP conjecture.

On the positive side, it was shown by Hastings~\cite{Has11} that the Toric Code satisfies
a variant of $\NLTS$ in which we only consider states whose fractional energy is vanishing
in $n$.
Namely, any quantum state for which $tr(H \rho) = \eps = o_n(1)$ can only be generated by circuits
whose depth is $\Omega(\log(1/\eps)) = \omega_n(1)$.  This implies, in particular, that the Toric Code
is also $\NLETS$ for any sub-constant function $\eps = \eps(n)$.

In proving the $\NLETS$ conjecture our proof will narrowly
dodge the above no-go theorems by employing a
Hamiltonian (a) on an expanding hypergraph, (b) albeit one with
$O(1)$ degree, (c) by being $O(1)$-local (although the terms do commute),
and (d) by having much smaller local expansion than one would expect
from a random graph, and specifically not being hyper-finite,  in the
language of \cite{Has12,FreedmanH13}.

We also review the various incomplete attempts at establishing NLTS:
\benum
\item The ``cat state'' is nontrivial but not the ground state of any
  local Hamiltonian~\cite{CJZZ12}.
\item Freedman and Hastings~\cite{FreedmanH13} give an example with
  ``one-sided NLTS.''  Here the Hamiltonian contains $X$ and $Z$ terms
  and has the property that a state which satisfies most of the $X$
  terms and all of the $Z$ terms must be nontrivial.  This result is also implied by the
  bounded-degree co-systole expanders of \cite{KKL14,EK16}.  
  In particular this one-sided $\NLTS$ is known {\it not} to be even $\NLETS$.
  
\item The uniform super-position over the code space of a classical code can be shown to
  require $\Omega(\log(n))$ depth to produce.  (This is a new result of ours but not hard
  to prove, and arguably was implicit in previous works ~\cite{BHF06,Haah14}.) If the code
  is an LDPC (low-density parity check) then this can be verified with $O(1)$-weight
  checks.  If it is an LTC (locally testable code; see \secref{qltc-nlts} for details),
  then this claim becomes robust, i.e.~it requires $\Omega(\log(n))$ depth to produce a
  distribution that even approximately matches the desired distribution.  However, there
  is no way with classical constraints to force the distribution to be uniform.  A single
  string in the support of the code space can be prepared in depth 1.  Only in a quantum
  code with both $X$ and $Z$ terms (or more generally, non-diagonal terms in the code
  Hamiltonian) can we require a state to be in a non-trivial superposition.  Thus the
  NLETS/cNLTS/NLTS conjectures do not have natural classical analogues, although our arguments
  for NLETS will rely on some techniques from classical coding.  
\eenum

\subsection{Discussion and Open Questions}
$\NLETS$ is a necessary condition for $\NLTS$ as well as $\qPCP$ \footnote{assuming
  $\QMA\subsetneq \NP$}, and $\qLTC$.  After a series of papers containing mostly bad news
for the $\qPCP$ and $\qLTC$ conjectures, this result can be seen as a pro-$\qPCP/\qLTC$
development.  We conjecture that $\NLTS$ is also true, and leave it as an immediate, yet
challenging open question.

Another open question is obtaining lower bounds on circuit depths that are asymptotically
larger than $\log(n)$.  Not only do our expansion lower bounds break down at this point,
but in fact all stabilizer states can be prepared in $O(\log(n))$ depth~\cite{AG04}.
Probabilistic arguments (cf.~\cite{BHH-designs}) can establish a loose depth hierarchy for
quantum circuits: circuits of depth $n^k$ cannot distinguish random circuits of depth
$n^{11k+9}$ from Haar-random circuits.  But it would be more useful to have concrete
methods for lower bounding depth at levels above $\log(n)$.

Finally, one can consider our construction in the context of $\qPCP$ itself.  
Of course our Hamiltonians do not encode any
particular problem and always have ground-state energy 0,
so in this context, our result should be regarded as a state-generation 
lower-bound and not as a complexity-theoretic result per-se.
Moving
beyond commuting Pauli operators and encoding
actual computational problems in the ground states of robust Hamiltonians
will be one of many hurdles required to
translate our result into a proof of the qPCP conjecture.

\section*{Acknowledgements}
We thank Dorit Aharonov, and Jeongwan Haah for useful conversations,
Fernando Brand\~{a}o, Toby Cubitt, Isaac Kim, Anthony Leverrier, Cedric Lin
and Kevin Thompson for useful and detailed comments and Alex Lubotzky for
pointing out the fundamental importance of the Tillich-Z\'{e}mor
construction.  
We thank an anonymous reviewer for pointing out a fundamental flaw in our previous paper
claiming the $\NLTS$ conjecture. 
We are both funded by the Templeton foundation, AWH
is funded by NSF grants CCF-1111382 and CCF-1452616 and ARO contract
W911NF-12-1-0486, and LE is funded by NSF grant CCF-1629809.

\newpage

\section*{Outline}

\renewcommand*{\contentsname}{}
\tableofcontents

\section{Preliminary facts and definitions}
\label{sec:prelim}

\subsection{Notation}\label{sec:notation}

\begin{itemize}
\item $\|M\|$ is the operator norm of $M$, i.e. its largest singular
  value.
\item For $x,y\in \bbF_2^n$, let $\dist(x,y)$ denote their Hamming
  distance, i.e.~the number of positions in which they differ.  
  For a point $x\in \bbF_2^n$ and a
  set $Y\in \bbF_2^n$ define $\dist(x,Y) := \min_{y\in
    Y}\dist(x,y)$.  For
  sets $X,Y\subseteq \bbF_2^n$, define 
  $$
  \dist(X,Y) = \min_{x\in X}\dist(x,Y).
  $$ 
  For $x\in \bbF_2^n$, $|x|$ is the Hamming weight of $x$.
\item A probability distribution $p$  on $\bbF_2^n$ is $(\eta, D)$-approximately
  partitioned if there exist sets $S_1,S_2$ with $\dist(S_1,S_2)\geq
  D$ and $p(S_1)\geq \eta, p(S_2)\geq \eta$.  We write simply {\em
    approximately partitioned} when $\eta= \Omega(1)$ and $D = \Omega(n)$.
\item
For subsets $A,B\subseteq \bbF_2^n$, let $A+B$ denote the set of all possible pairwise
sums $x+ y$ with $x\in A, y\in B$.
In particular, when $B$ has one element $x$, we may omit the set notation and write $A + x$.
\item
For a linear subspace $A\subseteq \bbF_2^n$ over $\bbF_2$, its dual $A^{\perp}$ is
$$
A^{\perp} = 
\left\{
x\in \bbF_2^n,
 \langle w,x \rangle =0 (\mod 2), 
\forall w\in A 
\right\}.
$$
We say that $A\perp B$ if $A \subseteq B^\perp$ or equivalently if
$B\subseteq A^\perp$.
\item
Let $C\subseteq \bbF_2^n$ be some code on $n$ bits. The minimal distance of $C$, denoted by $\Delta_{\min}(C)$
is the minimal distance between any pair of unique words in the code
$$
\Delta_{\min}(C) = \min_{x\neq y, x,y\in C} \dist(x, y).
$$
In this paper we study linear codes where $C$ is a subspace and so
$$
\Delta_{\min}(C) =\min_{x\in C-\{0\}}|x|.
$$
We also define the {\it fractional} minimal distance of $C$ by $\delta_{\min}(C) =
\Delta_{\min}(C)/n$. 
\item
  Let ${\cal C}$ be a linear code on $n$ bits, defined by the hypergraph $G = (V,E)$ where $V$
  corresponds to the set of bits/vertices and the $E$ corresponds to the set of
  checks/hyperedges.  We also define the $\bbF_2$-linear map $\partial$ from
  $\bbF_2^V\mapsto \bbF_2^E$ which sends a vertex $v$ to the sum over all $e\in E$ that
  are incident upon $v$. 
  The transpose of the code ${\cal C}$, denoted by ${\cal C}^T$, is defined by exchanging
  the roles of the bits and checks and replacing $\partial$ with $\partial^T$.   Note that
 $\calC = \ker\partial$ and $\calC^T = \ker\partial^T$.
We will find it convenient to overload notation so that
  $\partial v$ denotes also the set of edges $e$ incident upon $v$, and $\partial^Te$ is
  the set of vertices in hyperedge $e$.  
\item
Let $S\subseteq T\subseteq \bbF_2^n$ denote some linear subspaces.
Then $T / S$ denotes the quotient space, meaning the set of cosets
$\{t + S : t\in T\}$.
Also, $T - S$ denotes the set of strings in $T$ that are not in $S$.
In particular, $T/S$ has a representation in terms of elements of $T-S$, and the $0$ element (representing $S$).
\item
For two positive-semidefinite matrices $A \succeq 0,B \succeq 0$ we say $A \succeq B$ if $A-B \succeq 0$.
\item
For a set $S$ we denote by $U[S]$ the uniform distribution on the set $S$.
\end{itemize}

\subsection{Quantum codes and local Hamiltonians}
\label{sec:def-code}

\begin{definition}\textbf{Pauli operators}
\be X = \begin{pmatrix} 0 & 1 \\ 1 & 0 \end{pmatrix}
\qand Z = \begin{pmatrix} 1 & 0 \\ 0 & -1 \end{pmatrix}\ee
For $e\in \bbF_2^n$, define $X^e = X^{e_1} \ot X^{e_2} \ot \cdots \ot
X^{e_n}$, i.e.~the tensor product of $X$ operators in each position
where $e_i=1$; similarly define $Z^e = \bigotimes_i Z^{e_i}$.
\end{definition}

\begin{definition}

\textbf{CSS code}

\noindent
A $[[n,k,d]]$ quantum CSS code on $n$ qubits 
is a subspace ${\cal C}\subseteq {\cal H} = (\bbC^2)^{\ot n}$ of $n$ qubits. 
It is defined by a pair of  linear subspaces of $S_x, S_z \subseteq
\bbF_2^n$ such that $S_x \perp S_z$.
It is thus denoted  ${\cal C} = {\cal C}(S_x,S_z)$. 
Explicitly the subspace is given by
\ba
\calC(S_x, S_z) &=
\left\{\ket\psi\in (\bbC^2)^{\ot n} : X^x\ket\psi = \ket\psi \,\forall
x\in S_x,
Z^z\ket\psi = \ket\psi \,\forall z\in S_z\right\}
\\ & 
=\Span\left\{\frac{1}{\sqrt{|S_x|}} 
\sum_{x \in S_x} \ket{z+x} : z \in S_z^{\perp}\right\}.\ea
The code has $k = \log(|S_x^{\perp}/S_z|)$ logical qubits and distance $d = \min_{w\in S_x^{\perp}
  - S_z, S_z^{\perp} - S_x} |w|$. 
\end{definition}
The spaces of logical $X,Z$ operators are 
respectively defined by the quotient spaces
$S_z^{\perp}/S_x, S_x^{\perp}/S_z$.
The logical $X,Z$ operators that perform non-identity operations (also known as nontrivial
logical operators) are given by
$S_z^\perp-S_x$,$S_x^\perp-S_z$, respectively.

\begin{definition}\label{def:lh}

\textbf{$k$-local Hamiltonian}

\noindent
A $k$-local $n$-qubit Hamiltonian $H\succeq 0$
is a positive semidefinite operator on the $n$-qubit space $(\bbC^2)^{\otimes n}$,
that can be written as a sum $H = \frac{1}{m}\sum_{i=1}^m H_i$, where each $H_i$ is a positive-semidefinite matrix,
$0\preceq H_i \preceq I$, and each $H_i$ may be written as $H_i = h_i \otimes
I$, where $h_i \succeq 0$ is a $2^k \times 2^k$ PSD matrix.  
\end{definition}
These choices of eigenvalue bounds are to some extent arbitrary, but are also designed to
set the scale so that we can define ``low-energy'' states below in a natural way.

\begin{definition}
\textbf{The Hamiltonian of the code} 

\noindent
Suppose ${\cal C} = \calC(S_x, S_z)$ is a CSS code and
$H_x,H_z$ are subsets of $\bbF_2^n$ that generate
$S_x,S_z$.
Then we can define a Hamiltonian $H({\cal C})$, 
whose terms correspond to the 
generators of the CSS code in the following way.
\be
H = H({\cal C})
= \frac{1}{2|H_x|} \sum_{e\in H_x} \frac{I+X^e}{2}
+ \frac{1}{2|H_z|} \sum_{e\in H_z} \frac{I+Z^e}{2}
\label{eq:code-Ham}.\ee
\end{definition}
\noindent
Observe that the CSS condition $S_x \perp S_z$
implies that the
terms of $H({\cal C})$ all commute.  Thus the ground subspace of
$H(\calC)$ is precisely the code-space $\calC$.  Moreover, 
if the generating sets $H_x,H_z$ contain only terms with weight $\leq
k$ then the corresponding 
Hamiltonian $H({\cal C})$ is a $k$-local Hamiltonian.  

We can think of
$H$ as checking whether a state is a valid code state with the energy
equal to the expected fraction of violated constraints.
  However, in general, the number of violated
constraints may not correspond to more conventional notions of
``distance,'' such as (for classical strings) the Hamming distance to
the nearest codeword.  In the next section we discuss a type of code that addresses this.

\subsection{Locally Testable Codes}

\begin{definition}\label{def:ltc}
\textbf{Classical locally testable code}

\noindent
A code $C\subseteq \bbF_2^n$ is said to be locally testable with
soundness $\rho$ and query $q$, if there exists a set of $q$-local check terms
$\{C_1,\ldots,C_m\}$, such that
$$
\Prob_{i\sim U[m]} \left[C_i(w) = 1\right] \geq \rho \cdot \frac{\dist(w, C)}{n}.
$$
In particular $w\in C$ iff $C_i(w)=0$ for all $i$.
\end{definition}

\noindent
Similarly, a quantum locally testable code can be defined by the
property that quantum states at distance $d$ to the codespace have energy
$\geq \Omega(d/n)$.  (This normalization reflects the fact that the
check Hamiltonian $H(\calC)$ has norm $\leq 1$.)  Our definition is a
slight variant of the one from \cite{AE15}.
\begin{definition}\label{def:subspace-distance}
If $V$ is a subspace of $(\bbC^2)^{\ot n}$ then define its
$t$-fattening to be 
\be V_t := \Span \{ (A_1 \ot \cdots \ot A_n)\ket\psi : \ket\psi \in V,
\# \{i : A_i \neq I\}\leq t\}.\ee
Let $\Pi_{V_t}$ project onto $V_t$.  Then define the distance operator
\be D_V := \sum_{t\geq 1} t (\Pi_{V_t} - \Pi_{V_{t-1}}).\ee
\end{definition}
This reflects the fact that for quantum states, Hamming distance
should be thought of as an observable, meaning a Hermitian operator
where a given state can be a superposition of eigenstates.

\begin{definition}\label{def:qltc}

\textbf{Quantum locally testable code}

\noindent
An $(q,\rho)$-quantum locally testable code ${\cal C}\subseteq (\bbC^2)^{\ot n}$, 
is a quantum code with $q$-local projection $C_1,\ldots,C_m$
such that 
\be
 \frac{1}{m}\sum_{i=1}^m C_i 
 \succeq   \frac{ \rho}{n}   D_{\calC}.
\ee
\end{definition}

For stabilizer codes (which we will study exclusively in this paper)
this can be seen to be equivalent to the 
definition in \cite{AE15}.  However we believe it gives a more
generalizable definition of quantum Hamming distance.

\noindent
We now state the following connection (proved in Claim 3 of
\cite{AE15}) between classical and quantum CSS locally testable codes: 
\begin{fact}\label{fact:CSS}
\textbf{Classical codes comprising a $\qLTC$ CSS code must also be locally testable}

\noindent
Let ${\cal C}(S_x,S_z)$ be a quantum CSS code corresponding to two linear codes
$S_x^\perp,S_z^\perp \subseteq \bbF_2^n$.  If ${\cal C}$ is a $(q,\rho)$-$\qLTC$
then $S_x^\perp$ and $S_z^\perp$ are each
$q$-$\LTC$s with soundness at least $\rho/2$.  Conversely, if
$S_x^\perp$ and $S_z^\perp$ are each
$q$-$\LTC$s with soundness parameter $\rho$ with $S_x\subseteq S_z^\perp$ then
$\calC(S_x,S_z)$ is a $(q,\rho)$-$\qLTC$.
\end{fact}

We now present a slight re-wording of the definition of $\LTC$ which would be useful later on:
\begin{fact}\label{fact:cluster} 

\textbf{The words of a residual LTC cluster around the original code}

\noindent
Let $C$ be a locally testable code with parameter $\rho$.
Any word $w$ that violates a fraction at most $\eps$ of the checks of $C$ is at fractional
distance at most $\eps/\rho$ from $C$.
\end{fact}

\subsection{Expander Graphs}

Expander graphs are by now ubiquitous in computer science, and have been shown to be a
crucial element for many complexity theoretic results, most prominently, perhaps is the
combinatorial version of the PCP theorem~\cite{Din07}.  The term ``expander'' (or more
precisely ``edge expander'') refers to the fact that for a not-too-large subsets $S$ of
vertices a large fraction of the edges incident upon $S$ leave $S$.  (In
\secref{expansion} we will discuss the related but inequivalent phenomenon of vertex expansion.)
Formally, we define a discrete analogue of the isoperimetric constant, known as the
Cheeger constant
\be
h(G) = \min_{S\subseteq [n], 0 < |S| \leq n/2} \frac{|\partial_G(S)|}{|S|},
\ee
where $\partial_G(S)$ is the set of edges with one point in $S$ and one in $V - S$.
\begin{definition}
\textbf{Expander Graphs}
\noindent
A family of $d$-regular graphs $\{G_n\}_n$ is said to be expanding, if there exists a constant $h>0$
such that $h(G_n) \geq h$ for all sufficiently large $n$.
\end{definition}

Expanders can be defined equivalently in terms of the spectrum of their adjacency matrix.
For a graph $G$ we define $\lambda_2(G)$ as the second eigenvalue of its adjacency
matrix.  Since the top eigenvalue of the adjacency matrix of a $d$-regular graph is $d$,
we define the spectral gap to be $d-\lambda_2(G)$.  We say a family $\{G_n\}$ is
[spectrally] gapped if $\lim\inf d-\lambda_2(G_n)>0$.
It is was shown by Tanner, Alon and Milman that spectrally gapped graphs have a large Cheeger constant:
\begin{fact}\label{fact:cheeger}\cite{HooryLW06}
For any $d$-regular graph $G$ we have:
\be
h(G) \geq \frac{1}{2}(d - \lambda_2).
\ee
\end{fact}

Following the seminal results of Magulis and Lubotzky, Philips and Sarnak it is known that there exist
infinite families of expander graphs of degree $d = O(1)$, with $\lambda_2 \leq 2 \sqrt{d-1}$.
These are called Ramanujan graphs.
\begin{definition}\label{def:ramanujan}

\textbf{Ramanujan graphs}

\noindent
A family of graphs $\{G_n\}_n$ is said to be Ramanujan, if $\lambda_2(G_n) \leq 2 \sqrt{d-1}$
for all sufficiently large $n$.
\end{definition}

\noindent
In \cite{LubotzkyPS88,Mor94} it was shown how to construct Ramanujan graphs explicitly:
\begin{fact}\label{fact:explicit}
There exists an explicit infinite family of $d$-regular Ramanujan graphs for every $d = q+1$, where $q$ is a prime power.
\end{fact}
\noindent
We will use these expander graphs to construct an $\NLETS$ local Hamiltonian
in Section \ref{sec:nlets}.  (In fact, we do not strictly need Ramanujan graphs, but will
use specifically graphs where $h(G_n)\geq 3$.  Ramanujan graphs are simply a convenient
way to achieve this.)

Expander graphs of bounded-degree give rise naturally to {\it locally-testable codes} as follows.
Given an expander graph $G = (V,E)$ we define the following code $C(G)$.
It is the repetition code on $|V|$ bits,
with equality constraints of the form $x_i \oplus x_j = 0$ for all $(i,j)\in E$.
One can easily check that this code $C(G)$ is locally testable.
\begin{fact}\label{fact:exp2ltc}

\textbf{The repetition code from expander graphs is $\LTC$}

\noindent
The code $C(G)$ is the repetition code with a set of checks that is locally testable with query size $q = 2$, and soundness 
$\rho = 2h(G)/d$.
In particular, for $d$-regular Ramanujan graphs we have $\rho \geq 1 - \frac{2  \sqrt{d-1}}{d}$.
\end{fact}

In this paper we require a slightly more robust version of this fact where we allow
the adversarial removal of a small fraction of the vertices and edges.
\begin{definition}\label{def:res1}

\textbf{Maximal-connected residual graph}

\noindent
Let $G = (V,E)$, and subsets $V_{\eps}\subseteq V, E_{\eps}\subseteq E$.
A connected residual graph of $G$ w.r.t.~these sets is a graph $G' = (V',E')$ where
$V' \subseteq V_{\eps}, E'\subseteq E_{\eps}$ such that $G'= (V',E')$
is connected.
A maximal-connected residual graph $G_{\eps}$ is a connected residual graph
of maximal size $|V'|$. 
\end{definition}

Using the expander mixing lemma, it is easy to check that if $G$ is a $d$-regular Ramanujan
graph, then for sufficiently small constant $\eps>0$ there exists sufficiently large $d = O(1)$,
such that for any $V_{\eps}, E_{\eps}$ there exists a maximal-connected residual graph $G_{\eps} = (V',E')$ 
of large size
\begin{fact}\label{fact:eml1}
  Let $G=(V,E)$ be a $d$-regular Ramanujan graph, and let $\eps>0$.
For any $V_{\eps}, E_{\eps}$, where 
  $|V_\eps|\geq (1-\eps)|V|, |E_{\eps}| \geq (1-\eps)|E|$, 
  any
  corresponding maximal-connected residual graph $G_{\eps} = (V',E')$ satisfies:
  $|V'| \geq (1 - \eps') |V|, |E'| \geq (1 -2\eps') |E|$, for
\be \eps' = \frac{\eps(d+1)}{1-\frac{2\sqrt{d-1}}{d}}. \label{eq:eps-prime-def}\ee
\end{fact}

Later we will choose $d=14$, in which case \eq{eps-prime-def} simplifies to 
\be \eps'\leq 31\eps \label{eq:d=14}\ee

\begin{proof}
Let $E_{\eps}' \subseteq E_{\eps}$ denote the subset of edges of $E_{\eps}$ incident on $V_{\eps}\times V_{\eps}$.
By regularity of $G$ we can upper-bound 
$$
|E_{\eps}'| \geq |E| (1 - (d+1)\eps).
$$
Consider the graph $G' = (V, E_{\eps}')$ - any connected sub-graph of $G'$
is by definition a connected residual graph of $G_{\eps}$.
Let $S\subseteq V$ be a maximal connected component in $G'$.
Then $E_{\eps}'(S,\bar{S}) = 0$, 
implying that 
\be |E(S,\bar{S})| \leq \eps(d+1)|E|.\label{eq:ESS-ub}\ee
Let  $|S| := (1-\alpha) n$.  Then by \defref{ramanujan} and \factref{cheeger}, we have
\ba
|E(S,\bar{S})| & \geq |\bar S| \frac{d-2\sqrt{d-1}}{2}\\
\frac{|E(S,\bar{S})|}{|E|} &\geq \alpha \left(1-\frac{2}{\sqrt d}\right)
\label{eq:edges-leaving}\ea
Combining \eq{ESS-ub} and \eq{edges-leaving} implies that $\alpha \leq \eps'$, for $\eps'$
defined in \eq{eps-prime-def}.

Now define $V' = S$, and define $E'$ to be the subset of the edges $E_{\eps}$ incident on
$S$.
Then $G' = (V',E')$ is a connected graph with $|V'| \geq (1 - \eps') |V|$
and $|E'| \geq (1 - \eps'-\eps(d+1))|E|$.
\end{proof}

We use this property to derive the following:
\begin{proposition}\label{prop:subgraph}

\textbf{Robust LTC}

\noindent
Let $G = (V,E)$ be a $d$-regular Ramanujan graph.  Let $G_{\eps}= (V',E')$ denote a {\it
  maximal connected} residual graph of $G$ induced by a subset $V_{\eps}\subseteq V$,
$|V_{\eps}| \geq (1-\eps) |V|$ and $E_{\eps}\subseteq E$, $|E_{\eps}| \geq (1-\eps) |E|$.
Then any assignment $w'$ to the vertices of $V'$ violating at most $\delta$ fraction of the
checks of $G_{\eps}$ 
is $\frac{\delta + 2\eps'}{(1-\frac{2\sqrt{d-1}}{d})(1-\eps')}$-close to either
$\mathbf{1}_{V'}$ or $\mathbf{0}_{V'}$.

\end{proposition}
\begin{proof}
  Any word $w'$ defined on $V'$ that violates at most $\delta$ fraction of the checks of
  $E'$ can be extended with $0$'s to a word $w$ defined on $V$ that violates a fraction at
  most $\delta + 2\eps'$ of the checks of $E$, with $\eps'$ defined in \eq{eps-prime-def}.
  Since $G$ is Ramanujan then from \factref{exp2ltc}, $w$ is at fractional distance at
  most $\frac{\delta + 2\eps'}{1-\frac{2\sqrt{d-1}}{d}}$ to either $\mathbf{1}_V$ or $\mathbf{0}_V$ on $V$.  The
  fractional distance of $w'$ to either $\mathbf{1}_{V'}$ or $\mathbf{0}_{V'}$ can be larger by a factor of
  $|V|/|V'| \leq 1/(1-\eps')$, which implies that $w'$ is $\frac{\delta +
    2\eps'}{(1-\frac{2\sqrt{d-1}}{d})(1-\eps')}$ 
  close to either $\mathbf{1}_{V'}$ or $\mathbf{0}_{V'}$.
\end{proof}

\noindent
From the above one can derive the following corollary.
\begin{corollary}\label{cor:cheeger1}
  Consider the maximal connected residual graph $G_{\eps}$ above.  For
  all $d\geq 14$ and $\eps'$ defined in \eq{eps-prime-def}, the following holds: \be \forall w\in
  \F_2^{V'}, \quad 100 \eps' \leq \frac{|w|}{|V'|} \leq \frac 12 
\quad\Rightarrow\quad |\partial_{G_{\eps}} w|
  \geq 3 |w|.  \label{eq:3-expansion}\ee
\end{corollary}
\begin{proof}
This follows from \propref{subgraph}.  We set $\frac{|w|}{|V'|} =
\frac{\delta + 2\eps'}{(1-\frac{2\sqrt{d-1}}{d})(1-\eps')}$ (with $\eps'\leq
31\eps$ from \eq{eps-prime-def}), solve for $\delta$ and (after some algebra) find
that $\delta \geq 0.46\frac{|w|}{|V'|}$.  
This calculation uses the fact that the LHS of \eq{3-expansion} is only possible if $\eps'\leq
1/200$.
Thus the number of violated edges is 
\be \geq \delta |E'|
\geq 0.46\frac{|w|}{|V'|} (1-2\eps')\frac{dn}{2} \geq 3.15 |w|.\ee
\end{proof}

\section{Local Hamiltonians with Approximation-Robust
  Entanglement}\label{sec:entproxy}  
\cftsectionprecistoc{Precise statement of main results}

First, we will
precisely define our model of quantum circuits.  The following
definition codifies some of the common-sense features of circuits
that we will use.
\begin{definition}[Circuits]\label{def:circuits}
A (unitary) quantum circuit $C$ on $n$ qubits of depth $d$
is a product of $d$ layers $U_1,\hdots, U_d$,
where each layer $U_i$ can be written as a tensor-product
of $2$-local unitary gates $U_{i,(j,k)}$
\be
U_i = \bigotimes_{(j,k)\in P_i} U_{i,(j,k)},
\ee
where $U_{i,(j,k)}\in U(4)$ and each $P_i$ is a (possibly incomplete) partition of $[n]$
into blocks of size 2. 
\end{definition}
Corresponding to a circuit $C$ is a unitary
operator $U\in U(2^n)$ representing its action on an input state; often we will
simply refer to $U$ as a circuit when there is no ambiguity.

Low-depth circuits generate a family of ``simple'' states, known also as trivial states (\cite{FreedmanH13}).
\begin{definition}
\textbf{Depth-$d$ Trivial States (restated)}

\noindent
We say that an $n$-qubit state $\rho$ is depth-$d$ trivial if it can be prepared
by applying a depth-$d$ quantum circuit to $\ket{0}^{\ot N}$ (for some
$N\geq n$) and tracing out $N-n$ qubits. 
\end{definition}
An infinite family ${\cal F} = \{\rho_n\}_{n}$ of quantum states is said to be trivial
if there exists a constant $d$ such that $\rho_n$ is depth-$d$ trivial for all sufficiently large $n$.
We now define a family of NLTS Hamiltonians using the notion of trivial states as follows:

\begin{definition}\label{def:nlts}

\textbf{No Low-Energy Trivial States (NLTS) (restated)}  

\noindent
An infinite family of local Hamiltonians $\{H_n\}_{n\in \mathbf{N}}$ 
is $\eps$-NLTS if for any $d$ and 
a depth-$d$ trivial state family ${\cal F}$
and all sufficiently large $n$
\be \tr[\rho_n H_n] >\lambda_{\min}(H_n)+ \eps.\ee
We say that $\{H_n\}_{n\in \mathbf{N}}$ is NLTS if it is $\eps$-NLTS
for some constant $\eps>0$.
\end{definition}

This definition was motivated, in part, to prevent the following form
of NP-approximation of the ground-state energy of such system: a prover
sends a (polynomial-size) description of the shallow quantum circuit,
and the verifier computes the expectation value of
the Hamiltonian, conjugated by this unitary circuit, on the all-zero
state.  The verifier is thus able to accept/reject correctly. 
Since the circuit has depth $O(1)$, and each term
of $H$ is local, each local term of $U H U^{\dag}$ is local, so this
computation can be carried out efficiently.  In general $\tr[\rho H]$
can be estimated in ${\rm DTIME}\left(2^{2^{O(d)}}\right)$ if $d$ is a
depth-$d$ trivial state, since it requires estimating observables on
neighborhoods of $2^{O(d)}$ qubits.  (Similar but more complicated
results hold when we replace a depth-$d$ circuit with $e^{-iH'}$ for
$H'$ a sum of local terms in which each qubit participates in
interactions with total operator norm $O(d)$.~\cite{BHF06,Osb06})

\begin{conjecture}[NLTS conjecture~\cite{FreedmanH13}]\label{conj:nlts}
There exists a family of $O(1)$-local Hamiltonians with the NLTS property.
\end{conjecture}

Our main result will be stated in terms of hard-to-approximate
classical probability distributions as follows.  Recall that $\QNC^1$ is the
set of languages computable in quantum bounded-error log depth.  We
will use the term to describe classical distributions that can be approximately simulated with a
quantum log-depth circuit.
\begin{definition}

\textbf{$\QNC^1$-hard distribution}

\noindent
A family of  distributions $\{\calD_n\}$ on $n$ bits is said to be
$\QNC^1$-hard if there exist constants $a,c>0$ such that for
sufficiently large $n$ any $n$-qubit
depth-$c\cdot\log(n)$ trivial state $\rho_n$ satisfies
\be \| \calD_n  - \diag(\rho_n) \|_1 = \Omega(n^{-a}).
\ee
\end{definition}
Here $\diag(\rho)$ can be thought of as the probability distribution
resulting from measuring $\rho$ in the computational basis.
Next, we define quantum states as $\QNC^1$-hard if the
classical distribution induced by their measurement is hard to simulate {\it quantumly}:
\begin{definition}

\textbf{$\QNC^1$-hard quantum states}

\noindent
A family of $n$-qubit quantum states ${\cal F} = \{\rho_n\}_n$ is said to be $\QNC^1$-hard
if the corresponding family of distributions $\calD_n = \diag(\rho_n)$ is
$\QNC^1$-hard.
\end{definition}

Now, we can define local Hamiltonians as $\QNC^1$-hard if their ground states are
$\QNC^1$-hard:
\begin{definition}

\textbf{$\QNC^1$-hard local Hamiltonian}

\noindent
A family of local Hamiltonians $\{H_n\}_n$ is said to be $\QNC^1$-hard
if any family of states ${\cal F} = \{\rho_n\}_n$, with $\rho_n\in \ker(H_n)$ is $\QNC^1$-hard.
\end{definition}

As the final step we define a robust version thereof where we ask 
that even ground-state {\it impostors} are hard: 
\begin{definition}

\textbf{$\QNC^1$-robust local Hamiltonian}

\noindent
A family of local Hamiltonians $\{H_n\}$ is $\QNC^1$-robust 
if there exists $\eps>0$ such that
any family ${\cal F} = \{\rho_n\}_n$, where $\rho_n$ is an $\eps$-impostor of $H_n$ for all sufficiently large $n$, is $\QNC^1$-hard.
\end{definition}

\noindent
We can now state our main result.
\begin{theorem}\label{thm:nlets-informal}
\textbf{$\NLETS$} (sketch)

\noindent
There exists a family of $O(1)$-local Hamiltonians that is $\QNC^1$-robust.
\end{theorem}

Most of the remainder of the paper is devoted to the proof of \thmref{nlets-informal}.
In \secref{expansion} we will prove that the probability distributions
resulting from low-depth circuits cannot be approximately
partitioned.  

Then we will show that the distribution resulting from measuring quantum code-states {\em can} be approximately
partitioned.  The canonical example of such a partition is the cat state, as we
mentioned in the introduction, and indeed it is well known that the
cat state cannot be prepared in sub-logarithmic depth.  In
\secref{warmup} we will prove a ``warm-up'' result showing that any
Hamiltonian corresponding to a CSS code with $n^{\frac 12 + \Omega(1)}$ distance
is $\QNC^1$-hard.

To find a $\QNC^1$-robust Hamiltonian we will need a CSS code with
stronger properties.  In \secref{qltc-nlts} we show that a qLTC with linear distance
(quantum locally testable code) gives rise to a $\QNC^1$-robust
Hamiltonian, and in fact, the $\NLTS$ property.  

While no qLTCs are known, these ideas provide a sense of
how our full proof works.  Our construction is described in
\secref{TZ-review}, where the Tillich-Z\'emor hypergraph product from
\cite{TZ09} is reviewed, and in Section \ref{sec:nlets} where we use
it together with classical LTCs to construct our family of codes.
We then prove that this family is $\QNC^1$-robust.

\section{The Uncertainty Lemma and Noisy Quantum Code-States}

We next present a version of the classic
uncertainty principle~\cite{Robertson} that implies that if two
logical operators of a CSS codes anti-commute any state must have a
high uncertainty (i.e.~variance) in at least one of these operators.
This ``sum'' version is due to Hoffman and Takeuchi~\cite{HT03}.

\begin{lemma}\label{lem:uncertainty}
Let $\ket{\psi}$ be a quantum state, and $A,B$ Hermitian observables satisfying
$AB+BA=0$ and $A^2=B^2=I$.
Define
$$
\Delta A^2 = \langle \psi | A^2 | \psi \rangle - \langle \psi | A | \psi \rangle^2.
$$
Then
\be 
\Delta A^2 + \Delta B^2 \geq 1.
\label{eq:uncertainty}
\ee

\end{lemma}

Since the proof is short and our assumptions are slightly different
from those of \cite{HT03}, we present a proof here.

\begin{proof}
Define the operator
$$C = \vev{A} A + \vev{B} B$$
where $\vev{X} := \bra\psi X \ket \psi$.
Define $\lambda \equiv{\vev{A}^2 + \vev{B}^2}$.  Then we can directly calculate
\be C^2 = \lambda I \qand \vev{C} = \lambda \ee
For any random variable $X$, $\bbE[X]^2 \leq \bbE[X^2]$.
Thus $\lambda =\vev{C} \leq \sqrt{\vev{C^2}} = \sqrt{\lambda}$, implying that $\lambda \leq 1$.
Together with the fact that $\vev{A^2}=\vev{B^2}=1$ this implies \eq{uncertainty}. 
\end{proof}

Next, we require a simple fact that any CSS code
has a pair of bases, one for each of the quotient logical spaces, that anti-commute in
pairs.  The proof 
can be found for example in \cite{NC11}.
\begin{fact}\label{fact:basis}

\textbf{Anti-commuting logical operators}

\noindent
Let ${\cal C}$ be a $[[n,k,d]]$-CSS code: $\calC = \calC(S_x,S_z)$.
There exist sets
\begsub{logical-bases}
 {\cal B}_x &= \{b^x_1,\ldots, b^x_k\} \subset S_z^\perp \\
{\cal B}_z  & = \{b^z_1,\ldots, b^z_k\} \subset S_x^\perp
\endsub
 such that $\{b_i^x + S_x\}_{i\in [k]}$ and $\{b_i^z + S_z\}_{i\in
   [k]}$ are bases for $S_z^\perp/S_x$ and $S_x^\perp/S_z$
 respectively and
\be \langle b^x_i,b_j^z\rangle =
\delta_{i,j}.\label{eq:conjugate-pair}
\ee
\end{fact}
Here we should think of $\{X^{b^x_i}\}$ and $\{Z^{b^z_i}\}$ as logical
$X$ and $Z$ operators.

One useful property of CSS codes is that the value of the logical operators can be read
off from measuring each qubit individually.  If we measure a code state of
$\calC(S_x,S_z)$ in the $Z$ (resp.~$X$) basis then the outcomes will always lie in
$S_z^\perp$ (resp.~$S_x^\perp$).  The $+1$/$-1$ eigenvalues of the first logical $Z$
operator $Z^{b_1^z}$ correspond to the outcomes $S_z^\perp \cap (b_1^z)^\perp$ and 
$b_1^x + S_z^\perp \cap (b_1^z)^\perp$ when measuring each qubit in the $Z$ basis.
Observe also that $S_z^\perp \cap (b_1^z)^\perp = (S_z \cup b_1^z)^\perp  = S_x +
\Span(\calB_x - b_1^x)$.
Let us define accordingly the sets
\begsub{CXZ}
 C_0^Z &= (S_z \cup b_1^z)^\perp
& C_1^Z &= b_1^x  + C_0^Z\\
 C_0^X &= (S_x \cup b_1^x)^\perp
& C_1^X &= b_1^z  + C_0^X
\endsub
The sets $C_0^Z,C_1^Z$ (resp.~$C_0^X,C_1^X$)
partition $S_z^\perp$ (resp.~$S_x^\perp$).
Let ${\cal D}^Z_{\psi}$ (resp.~$\calD^X_\psi$)
denote the
distribution on $\bbF_2^n$ induced by measuring $\ket{\psi}$ in the tensor $Z$ basis
(resp.~the tensor $X$ basis), and
define $\vev{M} := \bra\psi M \ket\psi$ for any operator $M$.
The above discussion implies that if $\ket\psi\in\calC$ then 
\begsub{logical-ZX-exp}
 \vev{Z^{b_1^z}} &= \calD^Z_\psi(C_0^Z) - \calD^Z_\psi(C_1^Z) \\
 \vev{X^{b_1^x}} &= \calD^X_\psi(C_0^X) - \calD^X_\psi(C_1^X)
\endsub

Next we argue that uncertainty in the logical operators translates into uncertainty of
measurement outcomes in either the $X$ or $Z$ product basis.

\begin{proposition}\label{prop:hbar1}

\textbf{Uncertainty for code-states in at least one basis}

\noindent
Let $(S_x,S_z)$ be a CSS code with $\calB_x,\calB_z$ as in \factref{basis}.  Let
$\ket{\psi}$ be a quantum code-state, and $D_{\psi}^X,D_{\psi}^Z$ be the distribution of
the measurement of $\ket{\psi}$ in the Pauli-$X$ or Pauli-$Z$ basis, respectively.  Then
at least one of the following equations must hold:
\begsub{uncertainty-code-state}
D_{\psi}^Z(C_0^Z)
&\in \left[\frac{1}{2} -\frac{1}{2\sqrt{2}},  \frac{1}{2} +\frac{1}{2\sqrt{2}}\right]
\\
D_{\psi}^X(C_0^X)
&\in \left[\frac{1}{2} -\frac{1}{2\sqrt{2}},  \frac{1}{2} +\frac{1}{2\sqrt{2}}\right]
\endsub
Since $D_\psi^P(C_0^P)+D_\psi^P(C_1^P)=1$ for $P=X,Z$ we could equivalently state
\eq{uncertainty-code-state} in terms of $C_1^Z$ and $C_1^X$.
\end{proposition}

\begin{proof}
According to Lemma \ref{lem:uncertainty} any state $\ket\psi$
will have 
$$1 \leq (\Delta X^{b^x_1})^2 + (\Delta Z^{b^z_1})^2
 = 2 - \vev{X^{b_1^x}}^2 - \vev{Z^{b_1^z}}^2.$$
and therefore either $|\vev{X^{b_1^x}}|$ or $|\vev{Z^{b_1^z}}|$ must
be $\leq 1/\sqrt{2}$.  Assume w.l.o.g.~(since the other case is
similar) that 
\be \left |\vev{Z^{b_1^z}}\right| \leq 1/\sqrt{2} .\label{eq:large-fluct}\ee

The result now follows from \eq{logical-ZX-exp}.

\end{proof}

In this paper, we will mostly consider noisy code-states, and not actual code-states.  We
will want to argue that even noisy code-states have an uncertainty property w.r.t.~the
original logical operators.  To do that we consider pairs of Voronoi cells, corresponding to
pairs of anti-commuting logical operators as in the proposition above.  In geometry,
``Voronoi cells'' take a set of seeds and partition a space into the regions that are
closer to one seed than any other.   In classical coding theory, we can likewise partition the set
of strings according to which code word they are closest to.  Equivalently Voronoi cells
are the preimages of the maximum-likelihood decoding map when a string is subjected to
independent bit flip errors.  

For quantum CSS codes we
will partition the measurement outcomes in the $X$ and $Z$ bases into
the following analog of Voronoi cells:
\begin{proposition}\label{prop:vor1}

\textbf{Generalized uncertainty for unitary decoding}

\noindent
  Let ${\cal C} = (S_x,S_z)$ be a $[[n,k,d]]$-CSS code
  and $C_0^Z,C_1^Z,C_0^X,C_1^X$ are defined as in \eq{CXZ}.
  Let $E_x,E_z$ be some set of errors that satisfies:
\begsub{error-sets}
  S_0^z &:= C_0^Z + E_z \\
  S_1^z &:= C_1^Z + E_z \\
  S_0^{z} \cap S_1^{z} &= \emptyset
\endsub
and similarly this holds for the sets $S_0^x, S_1^x$, defined in the
same way w.r.t.~$E_x$.  
Suppose further that
\be
   {\rm supp}(D_{\psi}^Z)\subseteq S_0^{z} \cup S_1^{z} \qand
   {\rm supp}(D_{\psi}^X) \subseteq S_0^{x} \cup S_1^{x}
\ee

Then there exists a constant $c_0>0.07$ such
that 
\be
  (D_{\psi}^Z(S_0^{z}) \geq c_0 \qand D_{\psi}^Z(S_1^{z}) \geq c_0)
  \quad {\rm or} \quad
  (D_{\psi}^X(S_0^{x}) \geq c_0 \qand D_{\psi}^X(S_1^{x}) \geq c_0).
\ee
\end{proposition}

\begin{proof}
Define a decoding map for $X$ errors ${\cal U}_{\text{Dec}}^X$ as follows:
\be
\forall e\in E, w\in S_x^{\perp} \qquad
  {\cal U}_{\text{Dec}}^X 
  (\ket{w + e}_1 \otimes \ket{0}_2)
  = \ket{w}_1 \otimes \ket{e}_2,
\ee
Similarly, define a decoding map ${\cal U}_{\text{Dec}}^Z$ for
$Z$ errors:
\be
\forall e\in E, w \in S_z^{\perp} \qquad
 {\cal U}_{\text{Dec}}^Z  
 (\ket{w + e}_1 \otimes \ket{0}_2) 
= \ket{w}_1 \otimes \ket{e}_2,
\ee
Since $S_0^{z} \cap S_1^{z} = \emptyset$ and $S_0^{x} \cap S_1^{x} =
\emptyset$ then these maps are well-defined and can be extended to
unitary operations. 
In addition, since 
${\rm supp}(D_{\psi}^Z)\subseteq S_0^{z} \cup S_1^{z}$ and
   ${\rm supp}(D_{\psi}^X) \subseteq S_0^{x} \cup S_1^{x}$
then the decoded state
\be\label{eq:convex}
\rho = \tr_{2}(\calU_{\text{Dec}}^Z \circ \calU_{\text{Dec}}^X \ket{\psi}_1 \otimes \ket{0}_2) 
\ee
is supported entirely in $\calC$.
Let ${\cal D}_{\rho}^X, {\cal D}_{\rho}^Z$ denote the distribution on $\F_2^n$ induced by measuring $\rho$ in the $X,Z$ basis, respectively.

By Proposition \ref{prop:hbar1} for any code-state $\ket{\phi}\in {\cal C}$ the distribution ${\cal D}_{\phi}^X$ has a measure at least $1/2 - 1/(2\sqrt{2})$ on both sets
$C_0^X$ and $C_1^X$
or 
${\cal D}_{\phi}^Z$ has at least that measure
on both sets
$C_0^Z$ and $C_1^Z$ (defined in \eq{CXZ}).

By \eq{convex} the distributions ${\cal D}_{\rho}^X, {\cal D}_{\rho}^Z$ are each a convex combination 
of corresponding distributions ${\cal D}_{\phi}^X,{\cal D}_{\phi}^Z$, for code-states $\ket{\phi}\in {\cal C}$.
Hence by the above at least one of ${\cal D}_{\rho}^X, {\cal D}_{\rho}^Z$ has a measure at least 
$$
c_0 := 1/2 \cdot (1/2 - 1/(2\sqrt{2})) > 0.07
$$ 
on both $C_0^X, C_1^X$
or on both $C_0^Z, C_1^Z$.
Since by definition of the decoders $\calU_{\text{Dec}}^X, \calU_{\text{Dec}}^Z$ each $x\in C_0^X$ (or $x\in C_1^X$) could only have come from some $x'\in S_0^x$, or some $x'\in S_1^x$,
and not from both,
(and the same for the $Z$ basis), this implies that
${\cal D}_{\psi}^X(S_0^x) \geq c_0$ 
and
${\cal D}_{\psi}^X(S_1^x) \geq c_0$,
or this holds for the $Z$ basis.
\end{proof}

\section{Vertex expansion bounds for low-depth circuits}
\label{sec:expansion}
\cftsectionprecistoc{Proves \thmref{vertex-informal}, showing that trivial states have high vertex expansion.}

As stated above, a central notion of this paper (following Lovett and Viola~\cite{LV}) is that distributions
over codewords of good codes look very different from the outputs of
low-depth circuits.  We will see in this section that these can be
distinguished by comparing the different values of vertex expansion
that they induce on a particular graph.

Consider $\bbF_2^n$ to be the vertices of a graph with an edge
between all pairs $x,y$ with $\dist(x,y)\leq \ell$.  If $\ell=1$ then
this is the usual hypercube, but we will be interested in $\ell\approx\sqrt{n}$.  For a set $S\subseteq \bbF_2^n$
define $\pell(S)$ to be the boundary of $S$, meaning points
in $S$ connected by an edge to a point in $S^c := \bbF_2^n - S$,
along with points in $S^c$ connected to a point in $S$.  In
other words
\be\pell(S) = \{x \in S : \exists y \in S^c, |x-y|\leq \ell\}
\cup
\{x \in S^c : \exists y \in S, |x-y|\leq \ell\}
.\ee
Let $p$ be a probability
distribution over $\bbF_2^n$. 
The $p$-weighted vertex expansion is defined to be
\be h_\ell(p) := \min_{S, 0 < p(S) \leq \frac 1 2 }
\frac{p(\pell(S))}{p(S)}.\ee

\ignore{ Define the conductance of $p$ to be
\be \phi(p) := \min_{\substack{S \subset \bbF_2^n \\ 0 < p(S) \leq
    1/2}}
\Pr[y \not\in S : x \in S, |x-y| \leq \ell],\ee
where $x,y$ are drawn from $p$. }

In this section we argue that the outputs of low-depth circuits have
high vertex expansion for a suitable value of $\ell$.  To get intuition for
this, we consider first the case of the uniform distribution over
$\bbF_2^n$.  Here Harper's Theorem~\cite{Har64} implies that
$h_{\ell}(U[\F_2^n])$ is $\geq \Omega(1)$ when
$\ell = \Omega(\sqrt n)$.  In fact it goes further and gives the
exact isoperimetric profile, meaning it calculates
$\min_{p(S)=\mu} p(\pell(S)\backslash S)$ for all $\mu$, and shows that this minimization
is achieved for the Hamming ball.
Similar bounds are known
for any product distribution $p$.

This can be extended to the case when $p$ is the output of a classical
depth-$d$ circuit
$C: \{0,1\}^m \mapsto \bbF_2^n$
 which accepts $m$ uniformly random input bits and fan-in, fan-out
both $\leq 2$. In this case each output bit depends on at most $2^d$
bits.  Let
$S\subset \bbF_2^n$, $T = C^{-1}(S)\subseteq \{0,1\}^m$, and $p = U[\F_2^m]$.  
Since the output can depend on $\leq n 2^d$ bits of the
input, we can assume without loss of generality that $n\leq m\leq n
2^d$, or if $d$ is constant then $m = \Theta(n)$.
Using, for example, Harper's Theorem, 
one can show that if $p(T)\leq 1/2$, $x$ is drawn uniformly
from $T$ and $z$ is drawn uniformly from the Hamming ball $\{|z|\leq
\sqrt{m}\}$ then $x+ z$ has $\Omega(1)$ probability of lying in
$T^c$.  
Now we can use the assumption that the circuit is low depth to
argue that 
\be \dist(C(x), C(x+ z)) \leq |z|2^d \leq \sqrt{m}2^d \leq
\sqrt{n} 2^{1.5d}.\ee
Since $C(x) \in S$ and $C(x+ z)\in S^c$ this implies that $C(x)
\in \pell(S)$ with $\ell = \sqrt{n} 2^{1.5d}$.  We conclude that
$h_\ell(p) \geq \Omega(1)$.  This argument is a restatement of
Fact 4 from \cite{E15} (correcting a missing factor of $\sqrt{n}$ there).

The main result of this section is that a similar bound also holds
for the output of low-depth {\em quantum} circuits.  We first
formalize the fact that in sufficiently low depth circuits not all input bits can
influence a given output bit.

\begin{definition}[Light cone]
Given a depth-$d$ quantum circuit $C$ on $n$ qubits 
we define a directed acyclic graph $G=(V,E)$ by considering
$d+1$ layers of $n$ vertices, 
\be
V = \{V_0,\hdots, V_{d}\},  |V_i| = n\ \  \forall t\in \{0,\ldots,d\},
\ee 
where $V_0$ is the set of input qubits, and $V_{d}$
is the set of output qubits.
Recall from \defref{circuits} that the set of two-qubit gates at time $t$ defines a partition
$P_t$.  Then 
define the edge set $E$ by
connecting for all $t\in [d], k,l\in [n]$ the vertex $V_{t-1,k}$ to $V_{t,l}$ if 
$(k,l)\in P_t$. 
For a subset $S\subseteq V_{d}$ of output qubits, the {\em light cone} of $S$ 
is defined as the set 
$L(S)\subseteq V_0$ of input qubits from which there exists a directed
path in $G$ to some vertex in $S$.  The ``blow-up'' $B = B(C)$ of the circuit $C$
is defined as:
$$
B = B(C) = \max_{v\in V} L(v).
$$
\end{definition}

For depth-$d$ circuits comprised of $k$-qubit gates, the
blow-up is at most $k^d$.  
If the gates are required to be spatially
local in $D$ dimensions then this is $\leq (ckd)^D$ for some universal
constant $c$ depending on the specific geometry.  In this paper we
mostly are interested in the case of constant-depth circuits of
two-qubit gates with unrestricted geometry, although our results hold
more generally.  

We now state the main theorem of this section:
\begin{theorem}\label{thm:vertex}
Let $\ket\psi = U\ket{0^N}$ for $U$ a circuit with blow-up $B$.  Let $p$ be
the probability distribution that results from measuring the first $n$
qubits in the computational basis; i.e.~
\be p(x) = \sum_{y\in \{0,1\}^{N-n}} |\braket{x,y}{\psi}|^2.
\label{eq:low-depth-p}\ee
Then for $\ell = \frac 14B\cdot (B n)^{1/2 - \gamma}$ with $\gamma\in [0,1/2]$ 
we have:
\be h_{\ell}(p) \geq \frac{1}{8} (nB)^{-2\gamma}.\label{eq:poly-bound}\ee
\end{theorem}
Our proof is inspired by the use of Chebyshev polynomials by Friedman
and Tillich~\cite{FT00} to relate the diameter of a graph to the
spectral gap of its adjacency matrix, as well as by
\cite{AradKLV12area} to show that ground states of 1-d gapped
Hamiltonians have bounded entanglement. 

\begin{proof}
For $S\subseteq \bbF_2^n$ let
$\chi_S(x)$ denote the characteristic function of $S$: it is $-1$ if $x\in S$ and $1$ if not.
Define
the reflection operator 
\be R = \sum_{\substack{x\in \bbF_2^n\\ y\in \{0,1\}^{N-n}}}
\chi_S(x) \proj{x,y}
.\ee
Now
define $\ket{\psi'} = R\ket{\psi}$.  To gain intuition, if $\ket{\psi}$ is analogous to
the cat state $(\ket{0^N} + \ket{1^N})/\sqrt 2$ then $\ket{\psi'}$
would be $(\ket{0^N} - \ket{1^N})/\sqrt 2$.  Our proof strategy will
be to construct an operator $K$ such that
\begsub{expansion-bound}
\bra\psi K \ket\psi & = 0 \label{eq:eb-zero} \\
\bra{\psi'}K\ket{\psi'} & \leq 4p(\pell(S)) 
\label{eq:eb-first}
\\
\bra{\psi'}K\ket{\psi'} & \geq \frac 12 (nB)^{-2\gamma} p(S) \label{eq:eb-second}
\endsub
Proving \eq{eb-first} will require that $K$ cannot detect the phase
flip far from the boundary, while proving \eq{eb-second} will require
that $K$ can nevertheless distinguish $\ket\psi$ from $\ket{\psi'}$.

Let $L\subseteq [N]$ denote the qubits in the light cone of $[n]$.
By the definition of blow-up, 
\be
|L| \leq nB.
\ee
Define the Hamiltonians 
\be\label{eq:h0}
H_0 = \frac{1}{|L|} \sum_{i\in L} \ketbra{1}_i \qand H =
UH_0U^\dag.
\ee
Note that $H_0$ can be thought of as the code Hamiltonian
(cf.~\eq{code-Ham}) for the subspace with $\ket{0}^{\otimes |L|}$ for the qubits in
$L$ and arbitrary states elsewhere.
Both $H_0$ and $H$ have all eigenvalues between 0 and 1, both have a $2^{N-|L|}$-fold
degenerate ground space and both have gap $1/|L|$ to the first non-zero eigenvalue.
Define $P_0$ to project onto the states that
are $\ket 0$ in each of the qubits in $L$, i.e.
\be P_0 = \proj{0}^{\ot L} \ot I^{L^c}.\ee
Similarly define $P = UP_0U^\dag$.
The idea of $H_0$ (resp.~$H$) is to approximate $I-P_0$
(resp.~$I-P$), and indeed they have the same 0-eigenspace.  However,
$H_0$ and $H$ have other eigenvalues as small as $1/|L|$, making this
a rather weak approximation.  We will obtain a better approximation by taking polynomials
of these operators, and will find that higher degree buys us a better
approximation.  Indeed $\lim_{m\rightarrow \infty}(I-H_0)^m = P_0$.  But our
proof will require the sharper degree/error tradeoff that comes from
using Chebyshev polynomials, which we will see allows a degree of approximately $\sqrt{|L|}$.

Let $m = \frac 12 (B n)^{1/2 - \gamma}$ so that $\frac{m^2}{L} \geq \frac
14(nB)^{-2\gamma}$. 
Following Lemma 4.1 of \cite{AradKLV12area} we let $T_m(x)$ denote the degree-$m$ Chebyshev
polynomial defined implicitly by the equation $T_m(\cos(x)) = \cos(mx)$ 
One can also write $T_m(x)  =\cos(m\cos^{-1}(x))$ for $|x|\leq 1$ 
and
$\cosh(m\cosh^{-1}(x))$ for $|x|\geq 1$.  Next define
\be C_m(x) := 1 - \frac{T_m(f(x))}{T_m(f(0))},
\qquad \text{where }
f(x) := \frac{1 + 1/|L| -2x}{1- 1/|L|}.\ee
Our choice of $C_m(x)$ guarantees that
\be C_m(0)=0,\label{eq:zero-fixed}\ee
and we have chosen $f$ so that it maps the interval $[1/|L|,1]$ to $[-1,1]$, and thus 
$C_m([1/|L|,1])$ takes values only in $[-1,1]$. 
This implies that
\be
C_m(x)   \geq 1 - \frac{1}{T_m(f(0))}
 \qquad \text{for }\frac{1}{|L|}\leq x \leq 1
\ee
To evaluate this, follow again Lemma 4.1  of \cite{AradKLV12area} to observe that
\ba T_m(f(0)) &=\cosh(m\cosh^{-1}(f(0))) \geq 
1 + \frac{(m\cosh^{-1}(f(0)))^2}{2}
\\
\cosh^{-1}(f(0)) &\geq 2\tanh\left(\frac 12\cosh^{-1}(f(0))\right) =
2\sqrt{\frac{f(0)-1}{f(0)+1}} = \frac{2}{\sqrt{|L|}} \\
C_m(x)  & \geq 1 - \frac{1}{1 + 2\frac{m^2}{|L|}} 
\geq 1 - \frac{1}{1 + \frac 12(nB)^{-2\gamma}}
\label{eq:Cmx-LB}\ea
From Taylor's Theorem $(1+x)^{-1} \leq 1-x+x^2$ for $x\geq 0$, and so \eq{Cmx-LB} yields
\ba
C_m(x)
\geq 
\frac{1}{2}
(nB)^{-2\gamma} - \frac{1}{4}(n B)^{-4\gamma} \geq \frac{1}{4} (nB)^{-2\gamma}
&\text{ for }\frac{1}{|L|}\leq x \leq 1
\label{eq:Cm-gap}\ea
Finally  we would like to bound $C_m(x)$ throughout its range.
\be 0\leq C_m(x) \leq 1 + \frac{1}{T_m(f(0))}\leq 2
\qquad \text{for } 0\leq x \leq 1.
\label{eq:Cm-range}\ee

Define $K_0 = C_m(H_0)$ and $K = UK_0U^\dag = C_m(H)$.  
Note that $K_0$ is $m$-local and $K$ is $\ell$-local for $\ell = Bm = \frac 14 B^{1.5-\gamma} n^{1/2 - \gamma}$.
Since $C_m(H_0)$ consists only of powers of $H_0$, it 
commutes with $H_0$ and can be evaluated by applying $C_m$ to each
eigenvalue of $H_0$; the same applies to $C_m(H)$ and $H$.
Thus
$0 \preceq K \preceq 2 I$.
By \eq{zero-fixed}, the 0-eigenvalues of $K$ are the same as the
0-eigenvalues of $H$, and from \eq{Cm-gap} and \eq{Cm-range}, all the other
eigenvalues of $K$ are between $\frac{1}{4}(nB)^{-2\gamma}$ and 2.  Thus we establish
\eq{eb-zero} as well as the
operator inequality
\be \frac{1}{4}(nB)^{-2\gamma} (I-P) \preceq K \preceq 2 (I-P) .\label{eq:KP-op-ineq}\ee

Now we proceed to compute the upper and lower bounds on
$\bra{\psi'}K\ket{\psi'}$ claimed in \eq{expansion-bound}.  Partition
$\bbF_2^n$ into four sets $S_1, S_2, S_3, S_4$ as follows:
\begsub{four-part}
S_1 &= (S\backslash \pell(S)) \times \{0,1\}^{N-n} \\
S_2 & =(S \cap \pell(S)) \times \{0,1\}^{N-n} \\
S_3 & =(S^c \cap \pell(S)) \times \{0,1\}^{N-n} \\
S_4 & =(S^c\backslash \pell(S)) \times \{0,1\}^{N-n}
\endsub
Decompose $\ket{\psi}$ accordingly as
$$\ket{\psi} = \ket{\psi_1} + \ket{\psi_2} + \ket{\psi_3} +
\ket{\psi_4},$$
where the $\ket{\psi_i}$ are sub-normalized states whose support is contained in
$S_i$.  
Note that  $p(\pell(S)) =
\| \ket{\psi_2}\|^2 + \| \ket{\psi_3}\|^2$.
Using this notation we can write
$$\ket{\psi'} = -\ket{\psi_1} - \ket{\psi_2} + \ket{\psi_3} +
\ket{\psi_4}.$$
Let $K_{ij} := \bra{\psi_i} K \ket{\psi_j}$. 

The fact that $K$ is
$\ell$-local means that $K_{13}=K_{14} = K_{24} = 0$.  Additionally
\eq{eb-zero} 
means that $\sum_{i,j \in [4]} K_{ij} = 0$.  Together, and since $K$ is Hermitian,
these mean that
\be \bra{\psi'}K\ket{\psi'} = -2K_{23} - 2K_{32} =
-4\text{Re}K_{23}.\ee
Since $\|K\|\leq 2$ we can use Cauchy-Schwarz to bound
\be  \label{eq:upper1}
|\bra{\psi'}K\ket{\psi'}| \leq 
8 \| \ket{\psi_2} \| \cdot \|\ket{\psi_3}\| \leq
4(\| \ket{\psi_2} \|^2 +  \|\ket{\psi_3}\|^2)
= 4 p(\pell(S)),\ee
thus establishing \eq{eb-first}.  

We now turn to the proof of \eq{eb-second}.  Observe that $U^\dag R U$
acts only on the qubits in $L$.  Thus 
$U^\dag R U\ket{0}^{\ot N} = U^\dag R\ket\psi
= U^\dag\ket{\psi'}$ is a superposition of states of the form
$\ket{x}^L \ot \ket{0}^{L^c}$, implying that
$$P_0 U^\dag \ket{\psi'} \propto \ket{0}^{\ot N}.$$
We can determine the proportionality constant by calculating
$$\bra{0}^{\ot N} P_0U^\dag \ket{\psi'} =
\bra{0}^{\ot N} U^\dag \ket{\psi'} 
= \braket{\psi}{\psi'} = \sum_{x\in \bbF_2^n} \chi_S(x) p(x) = 1-2p(S).$$
Thus $P_0U^\dag\ket{\psi'} = (1-2p(S))\ket{0}^{\ot N}$.  We can then
calculate
\begsub{UPU-prime}
\bra{\psi'}P \ket{\psi'}
& = \bra{\psi'}UP_0U^\dag \ket{\psi'} \\
& = \bra{\psi'}U(1-2p(S))\ket{0}^{\ot N} \\
& = (1-2p(S))\braket{\psi'}{\psi} \\
& = (1-2p(S))^2
\endsub
Finally we can bound
\begsub{K-prime}
\bra{\psi'}K\ket{\psi'} 
& \geq \frac{1}{4}(nB)^{-2\gamma} \bra{\psi'}(I-P)\ket{\psi'} & \text{using \eq{KP-op-ineq}} \\
& = \frac{1}{4}(nB)^{-2\gamma} (1 - (1-2p(S))^2)\\
& \geq \frac{1}{2}(nB)^{-2\gamma} p(S) & \text{using }p(S)\leq 1/2
\endsub
Taking now the ratio with the upper bound from \eq{upper1} implies
\be
h_l(p)=
\frac{p(\partial_l(S))}{p(S)} 
\geq
\frac{1}{8} (nB)^{-2\gamma}.
\ee

\end{proof}

Our proof that $\qLTC$s  of distance $\omega(\sqrt{n})$ are $\NLTS$ in Theorem \ref{thm:qltc} will require a slightly different graph property:
upper bounds on the distance between large sets instead of lower
bounds on the vertex expansion; i.e.~showing that $p$ cannot be
approximately partitioned.  The relation between these properties
is a standard fact that does not involve any features of quantum
circuits.

\begin{corollary}\label{cor:incomplete}
Let $p$ be a probability distribution on $n$ qubits generated by a
quantum circuit with blow-up $B$, as in \eq{low-depth-p}.  If $S_1,S_2\subset \bbF_2^n$
satisfy $p(S_1)\geq \mu$ and $p(S_2)\geq \mu$, then
$\dist(S_1,S_2)\leq 4\sqrt{n}B^{1.5}/\mu$.
\end{corollary}

The contrapositive of this claim is that if $p$ is 
$(\mu, D)$-approximately partitioned for $D>4\sqrt{n}B^{1.5}/\mu$ then it cannot result from
measuring the state resulting from a circuit with blowup at least $B$.  In terms of depth, $p$
cannot result from a circuit with depth at most
\be \leq \frac 23 \log\left(\frac{\mu D}{ 4\sqrt{n}}\right). \label{eq:depth-LB}\ee

\begin{proof}
Let $D = \dist(S_1,S_2)$ and $\ell = \sqrt{n}B^{1.5}/4$.  Assume $\mu\leq
1/2$, since otherwise we would have $D=0$.
For $t=0,1,2,\ldots$ define the sets
$$U_t := \{x : (t-1)\ell < \dist(x,S_1) \leq t\ell\}.$$
Then $S_1 = U_0$ and $S_2\cap U_t = \emptyset$ for all $t < D/\ell$.

This implies that 
$$\sum_{\substack{1 \leq t\leq \frac{D}{\ell}-2 \\ t \text{ odd}}} p(U_t) +
p(U_{t+1}) \leq 1 - 2\mu $$ and
in turn that there exists a particular $t_0$ for which 
\be p(U_{t_0}) + p(U_{t_0+1}) \leq
\frac{1-2\mu}{\frac{D}{2\ell}-1}
\label{eq:particular-t0}\ee
We will use $t_0$ to define a partition.  Let
$$\bar S_1 = \bigcup_{t\leq t_0} U_t
\qand
\bar S_2 =  \bigcup_{t> t_0} U_t.$$
Since $S_i \subseteq \bar S_i$ we have $p(\bar S_i)\geq \mu$ but now $\bar
S_1, \bar S_2$ form a partition of $\bbF_2^n$.  
Thus by \thmref{vertex}, $p(\pell(\bar S_1)) \geq \mu/8 $.
On the other hand $p(\pell(\bar S_1))\leq p(U_{t_0}) + p(U_{t_0+1})$. 
From \eq{particular-t0} we have
\be \frac{D}{2\ell} \leq 1 + \frac{1-2\mu}{\mu/8} \leq \frac 8\mu.\ee
\end{proof}

\corref{incomplete} applies to any value of $\mu$ but the bounds become weak when $\mu$ is
small or when $p(S_1),p(S_2)$ are not comparable.  One plausible generalization of
\corref{incomplete} would give upper bounds on $\dist(S_1,S_2)$ that scale optimally in
terms of both $p(S_1)$ and $p(S_2)$; see \cite{FT00} for the precise statement.  Since our
main results in this
paper do not make use of such a generalization we do pursue this further here.  

\section{Warm-up: 
\\ Quantum CSS Code-States are $\QNC^1$-hard}
\label{sec:warmup}
\cftsectionprecistoc{Non-robust circuit lower bounds for any CSS code}

Circuit lower bounds for generating quantum code states {\it exactly} can be readily
derived from the local indistinguishability property.  In this section, we show that our
techniques can be used to derive a robust version of this property, which is that quantum
CSS codes cannot be approximated by bounded-depth quantum circuits, even up to constant
$l_2$ error.  
Such robust versions have been shown for example by  Bravyi, Hastings
and Verstraete~\cite{BHV06}. 

Notably, even such a hardness-of-approximation claim is by no
means robust, because we still consider approximation of {\it perfect} ground states of
the code Hamiltonian.  In other words, while a code-state is $\QNC_1$-hard, not every
$\eps$-impostor of a code-state is $\QNC_1$-hard.
In fact, many constructions of quantum CSS codes are known to be {\it not}
$\QNC_1$-robust: i.e.~one can find $\eps$-impostors of such codes that are trivial (see \secref{previous}).

The claims in this section demonstrate our techniques by improving the approximation bounds on {\it perfect} code-states from $0$ error to constant $l_2$ error.
In subsequent sections we will then strengthen these bounds even further and demonstrate a quantum code
for which {\it every} $\eps$-impostor is $\QNC_1$-hard - which amounts to
the $\NLETS$ theorem.
Hence, the following claim on CSS codes relates only to code-states and not code-state impostors:
\begin{proposition}\label{prop:qecc-hard}

\textbf{Code-states of quantum CSS codes with large distance are $\QNC^1$-hard}

\noindent
Let ${\cal C} = [[n,k,\Delta_{min}]]$ be a quantum CSS code.
Preparing any $\ket\psi\in\calC$ up to $l_2$
 error at most $0.14$ requires depth $\Omega(\log(\Delta_{\min}/\sqrt n))$.
In particular, if $\Delta_{\min}\geq n^{1/2 + \Omega(1)}$ then $\ket\psi$ is $\QNC^1$-hard.
\end{proposition}
\begin{proof}[Proof of \propref{qecc-hard}]
Let $\ket{\psi}$ be some code-state of ${\cal C}$.
By Fact \ref{fact:basis} above, one can find bases ${\cal B}_x, {\cal
  B}_z$ satisfying \eq{conjugate-pair}.  Choose, say, the first pair
$b^x := b^x_1\in\calB_x$, $b^z := b^z_1\in\calB_z$.

Let $C_0$ denote the linear space $C_0 = S_z^\perp \cap (b^z)^\perp\subset \mathbf{F}_2^n$,
and define the affine space $C_1 = C_0 + b^x$.  
If $s_0\in C_0, s_1\in C_1$ then $s_0+ s_1 \in C_1
\subseteq S_z^\perp - S_x$, implying that $|s_0+ s_1|\geq
\Delta_{\min}$,
and so
\be \dist(C_0,C_1)\geq \Delta_{\min} \label{eq:Ci-far}.\ee

Let ${\cal D}_{\psi}^Z$ denote the
distribution on $\bbF_2^n$ induced by measuring $\ket{\psi}$ in the tensor $Z$ basis.
Then by Proposition \ref{prop:hbar1} we either have
\be \calD^Z_\psi(C_0) \geq \frac 1 2 - \frac 1 {2\sqrt{2}}
\qand \calD^Z_\psi(C_1) \geq \frac 1 2 - \frac 1 {2\sqrt{2}}
\label{eq:DC-lower-bound},\ee
or a similar statement holds for measuring in the $X$ basis.  WLOG assume that
\eq{DC-lower-bound} holds.
Thus ${\calD^Z_\psi}$ is approximately partitioned with measure at least
$\mu = \frac{1}{2}-\frac{1}{2\sqrt{2}}$ 
and distance $\Delta_{\min}$.
Hence, any distribution $p$ that is $\eps$-close to ${\calD^Z_\psi}$ for $\eps<\mu$ is 
$(\mu- \eps, \Delta_{\min})$-approximately partitioned.
Therefore, by \corref{incomplete} (and specifically \eq{depth-LB}) producing $\ket\psi$ to
error $\eps$ requires  depth
\be \geq \frac{2}{3} \log \left(\frac{(\mu-\eps)\Delta_{\min}}{4\sqrt n}\right).\ee
Since $\mu \geq 0.142\ldots$, if we take $\eps=0.14$ then this implies a
depth lower bound of $\frac 23 \log\frac{\Delta_{\min}}{\sqrt{n}} - O(1)$.  If
$\Delta_{\min}=n^{1/2+\Omega(1)}$ then this bound is $\Omega(\log n)$ and so $\ket\psi$ is
$\QNC^1$-hard. 
\end{proof}

\begin{figure}[h]
\center{
 \epsfxsize=5in
 \epsfbox{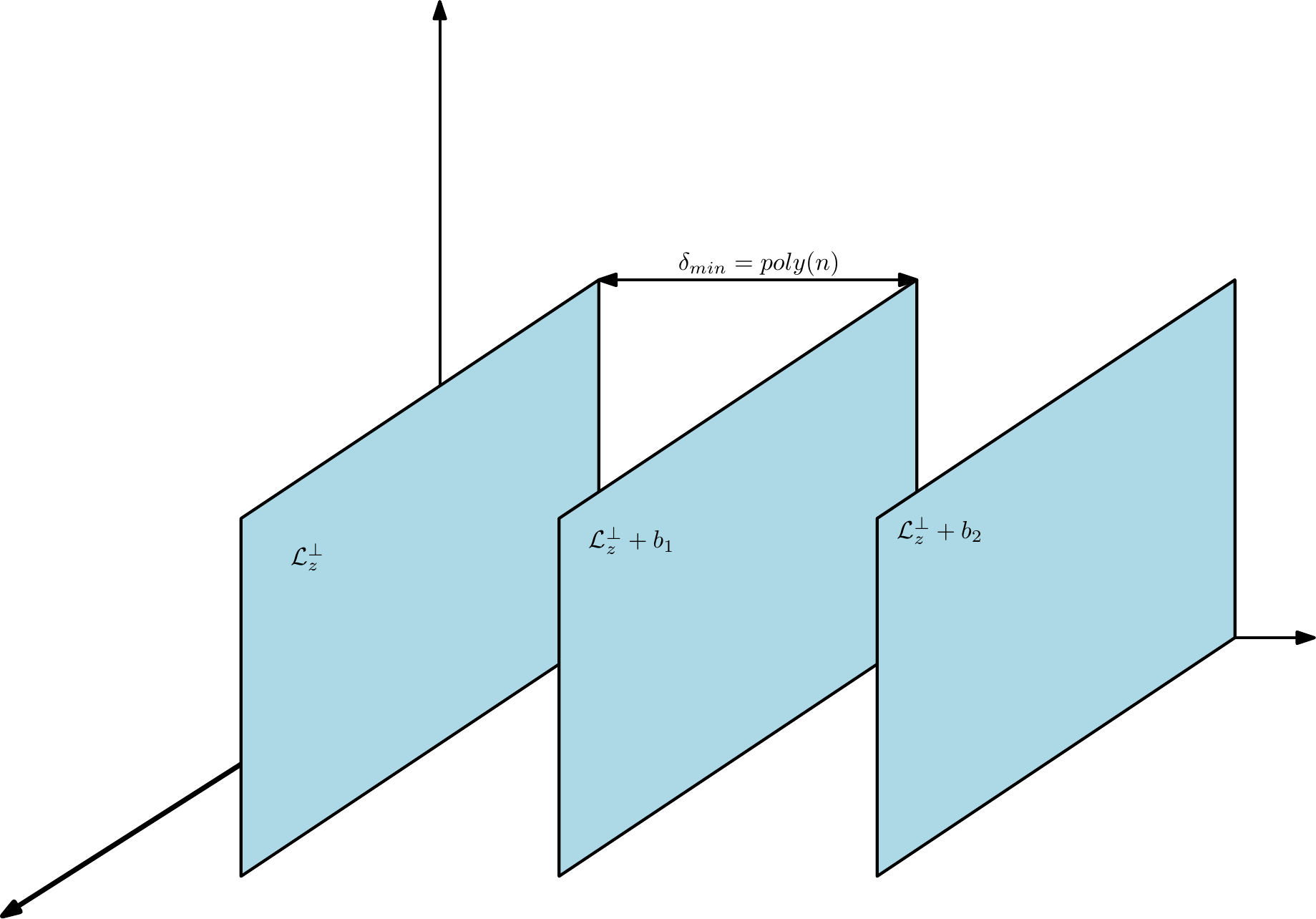}
 \label{fig:code}
 \caption{\small{Depiction of the approximate partition of a quantum CSS
     code with large distance.  Any code state must superpose non-negligibly in at least one of the two bases, on two distinct affine spaces separated by a large distance.}}}
 \end{figure}

\paragraph{Implications for known quantum codes}

\propref{qecc-hard} provides a nontrivial quantum circuit lower bound on the quantum LDPC
codes due to \cite{FML02}.  These codes are CSS codes and have distance
$\Omega(\sqrt{n\log(n)})$ which corresponds to a circuit depth lower bound of
$\Omega(\log\log(n))$.

One can also consider the toric code which has distance $\Theta(\sqrt{n})$.
In such a case, while one cannot apply directly \propref{qecc-hard} - one can still show
a similar proposition 
where the distribution is approximated to within distance $\eps =
n^{-1-\delta}$ for some constant $\delta>0$.   
\begin{proposition}
Let ${\cal C} = [[n,k,\Delta_{min}]]$ be a quantum CSS code with $\Delta_{\min}
\geq n^\alpha$ for $\alpha>0$ and $k\geq 1$. If $\ket{\psi}\in\calC$ and $\| \rho
- \proj\psi\| \leq n^{-1-\beta}$ for $\beta>0$ then preparing $\rho$ requires
depth $\Omega(\log(n))$.
\end{proposition}

We note that other methods are
known~\cite{BHF06,H-locality,Haah14} for showing that QECC
ground states, and even low-temperature thermal states of the 4-d
toric code~\cite{Has11}, are nontrivial.  Indeed our proof can be
viewed as a certain way of generalizing the argument of \cite{BHF06}.
Since the proof is very similar to that of \propref{qecc-hard} we provide only a brief
sketch of the proof here. 

\begin{proof}
Setting $\gamma=1/2$ in \thmref{vertex} yields\footnote{We observe that for these parameters the proof of \thmref{vertex} does not need to make use of any Chebyshev
polynomials and could simply set $K=H$.}
\be h_B(p) = \Omega(1/nB).\label{eq:weak-expansion}\ee

Next we follow the proof of \propref{qecc-hard} and construct the same
pair of sets $C_0,C_1$, with $D_\psi^Z(C_0),D_\psi^Z(C_1) \geq \Omega(1)$,
where $D_\psi^Z$ is the probability distribution resulting from
measuring $\ket\psi$ in the $Z$ basis.  Additionally observe that $D_\psi^Z(C_0\cup
C_1)=1$ and $\dist(C_0,C_1) \geq \Delta_{\min}$.

Let $p:=\diag(\rho)$.  By hypothesis 
\be \frac 1 2 \|p - D^Z_\psi\|_1 := \eps \leq n^{-1-\beta}.\label{eq:really-close}\ee
Suppose that $\rho$ can be generated by a circuit of 
depth $d$ and define $B=2^d$, so that \eq{weak-expansion} holds.  We will argue that \eq{really-close}
implies that $B$ is large.

If $B \geq \frac 12 \Delta_{\min}$ then we immediately have $d \geq \alpha\log(n)-1$.  
Otherwise, assume
$\dist(C_0,C_1)>2B$, and let
$S$ be the $B$-fattening of
$C_0$.  Then $D_\psi^Z(\partial_B(S))=0$, implying that
$p(\partial_B(S)) \leq \eps$.  On the other hand, $p(S)
= \Omega(1)$, so by \eq{weak-expansion}
we have $p(\partial_B(S))  = \Omega(1/nB)$.  Combining these we have
$B = \Omega(\frac{1}{n\eps})=\Omega(n^\beta)$, which again implies the $\Omega(\log(n))$ circuit
lower bound. 
\end{proof}

We remark that these arguments never made use of the LDPC property of
codes, and would apply equally well to say, a random stabilizer code
with linear distance and linear-weight stabilizers.  Rather the claims
are only nontrivial when applied to LDPC codes since it would not be
surprising if $n$-local stabilizers forced a system into an
$n$-partite entangled state.

\paragraph{Non-CSS codes.} It is tempting to speculate that \propref{qecc-hard}
should hold for any quantum code
(i.e.~not only CSS) with distance $\omega(n^{1/2})$ and at least one
logical qubit.  While we believe this is likely to be true, we would
need new techniques to prove it.  It is possible for such codes to
yield a nearly uniform distribution when measured in any local basis
(e.g.~consider a random 2-dimensional subspace of $(\bbC^2)^{\ot n}$)
which cannot be approximately partitioned.
We note that using a proof similar to that of \propref{qecc-hard} one can 
show $\QNC^1$-hardness for general (i.e.~non-Pauli) stabilizer codes, but we omit this here
for simplicity.

\section{A Bit Further: \\
Quantum Locally Testable Codes are $\NLTS$}
\label{sec:qltc-nlts}
\cftsectionprecistoc{Another warmup: qLTCs (which are not known to
  exist) would yield a proof of NLTS.}

We now connect between quantum locally testable codes (see definition of $\qLTC$'s in Definition \ref{def:qltc}) and local Hamiltonians with approximation-robust entanglement:
For a $\qLTC$ with large minimal distance, all low-energy states are
non-trivial.  This implies that the corresponding Hamiltonians have
the NLTS property.  
\begin{theorem}\label{thm:qltc}
Let ${\cal C} = [[n,k,\Delta_{min}]]$ be a quantum locally testable CSS code with soundness
$\rho>0$, $k\geq 1$, and 
$\Delta_{\min} = \Omega(n)$.
Then the local Hamiltonian of ${\cal C}$, $H({\cal C})$, is $\NLTS$.
\end{theorem}
However, no qLTCs with the required properties are
known. In fact, we do not even know of quantum LDPC codes
(i.e.~$O(1)$-weight check operators) with distance greater than
$O(\sqrt{n} \log^{1/4}(n))$.
Still the proof is
conceptually a bridge between the proof of the exact case in
\propref{qecc-hard} and our proof of NLETS in \thmref{nlets}.
\begin{proof}
Let $\ket{\psi}$ be some quantum state with:
\begin{equation}\label{eq:energy}
\langle \psi | H({\cal C}) | \psi \rangle \leq \eps,
\end{equation}
for $\eps>0$ a constant we will choose later.

From \factref{basis}, there exist a pair of logical operators $b^x_1,
b^z_1$.  Define $C_{0,1}^{X,Z}$ as in \eq{CXZ}.  We will apply
\propref{vor1} with error sets $E_x = E_z = \{w\in \F_2^n : |w|
< \frac 12 \Delta_{\min}\}$.  This implies the existence of sets
$S_0^x,S_1^x,S_0^z,S_1^z$ defined as in \eq{error-sets} such that 
such that $\ket{\psi}$ has projection at least $c_0$ on either both
$S_0^x$ and $S_1^x$ or both $S_0^z$ and $S_1^z$.
Assume w.l.o.g.~that
\be 
\calD_\psi^Z(S_0^z) \geq c_0
\qand \calD_\psi^Z(S_1^z) \geq c_0
\label{eq:D-tilde-C-LB}.
\ee

Now we cannot directly use the distance guarantees of the code since
$\ket\psi$ is not necessarily a code state.  However, by
\factref{CSS}, $S_x^\perp$ is locally testable with parameter
$\rho/2$.  Moreover, by \eq{energy} the expected fraction of violated
clauses in $H_Z$ is at most $2\eps$.  Thus 
\be \bbE_{z \sim \calD^Z_\psi} \dist(z, S_x^\perp)
\leq \frac{4\eps n}{\rho} 
\ee
So by a Markov argument:
\be
\calD^Z_\psi \left(\{z : \dist(z,S_X^\perp) \geq \frac{8\eps n}{ \rho c_0}\}\right) 
 \leq
c_0/2 \label{eq:mostly-close}
\ee
We now define the restriction of the Voronoi cells $S_0^x,S_1^x,S_0^z,S_1^z$ to the set of words that violate few parity checks:
\be \tilde{S}_i^z = \{z  \in S_i^z  : \dist(z, S_z^\perp) \leq \frac{8\eps n}{\rho c_0}\},\ee
for $i=0,1$.  By \eq{D-tilde-C-LB} and \eq{mostly-close}, we have for all $i\in \{0,1\}$,
$\calD^Z_\psi(\tilde{S}_i^z)\geq c_0 - c_0/2 = c_0/2$.  

On the other
hand, we claim that $\dist(\tilde{S}_0^z, \tilde{S}_1^z)$ is large:
Let $s_0\in \tilde{S}_0^z, s_1\in \tilde{S}_1^z$.
Then we can write $s_i = t_i + u_i$ with  $t_0, t_1 \in
S_z^\perp $ 
and $|u_i| \leq \frac{8\eps n}{\rho c_0}$.
In addition, 
since $s_0\in \tilde{S}_0^z, s_1\in \tilde{S}_1^z$ then 
their respective closest codewords are in {\it different} cosets modulo $S_x$, i.e.:
$t_0\in C_0^Z$ and 
$t_1\in C_1^Z$.
Hence by the minimal distance of the code: 
$\dist(t_0,t_1) \geq
\Delta_{\min}$. 
Then 
$\dist(s_0,s_1) \geq \Delta_{\min} - 16\frac{ \eps n}{\rho c_0}$.  
Thus
\be \dist(\tilde{S}_0^z, \tilde{S}_1^z) \geq \Delta_{\min} - 16 \frac{\eps n}{\rho c_0}.
\ee
Thus,
$D_\psi^Z$ is approximately partitioned with parameters each of which is at least $(c_0/2,
\Delta_{\min} - 16\frac{ \eps n}{\rho c_0})$.  For sufficiently small constant $\eps>0$,
\corref{incomplete} then implies an $\Omega(\log n)$ lower bound on
the depth required to prepare $\ket\psi$.
\end{proof}

\section{The Hypergraph Product}
\label{sec:TZ-review}
\cftsectionprecistoc{Review of a quantum coding technique used in the
  following sections.}

\subsection{General}
In this section, we survey the hypergraph product due to Tillich-Z\'{e}mor \cite{TZ09}.
We provide here only the very basic definitions that are required to prove our main
theorem, and refer the reader to the original paper~\cite{TZ09} for an in-depth view.  The
hypergraph-product code takes in two classical codes defined by their Tanner constraint
graphs and generates a product of these codes as hypergraphs.  Then it attaches a CSS code
to the product graph.  Formally stated:
\begin{definition}\label{def:TZ}

\textbf{The Hypergraph Product}

\noindent
Let $(V_1,E_1),(V_2,E_2)$ be two constraint hypergraphs with
corresponding edge-vertex incidence
operators $\partial_1,\partial_2$ and codes $\calC_1=\ker\partial_1,
\calC_2=\ker\partial_2$.
Then the Tillich-Z\'{e}mor hypergraph product 
of these codes, denoted by 
\be
\calC_{\times} = {\cal C}_1 \times_{TZ} {\cal C}_2,
\label{eq:TZ-notation}\ee
is defined by the hypergraph product of the corresponding graphs.
Specifically, its Hilbert space is comprised of qubits corresponding to
$$\left(V_1\times V_2\right) \cup \left(E_1\times E_2\right),$$
and check matrices are
\be  H_x = \left( \partial_1 \otimes I_{V_2} \middle| I_{E_1} \otimes \partial_2^T \right)
\qquad
 H_z = \left( I_{V_1} \otimes \partial_2 \middle| \partial_1^T \otimes I_{E_2}\right)
\ee
\end{definition}
These matrices have rows indexed by qubits and columns indexed by checks.  The
$X$ constraints, for example, are labeled by elements of $E_1\times V_2$, with 
constraint $(e_1,v_2)$ is connected to all elements $(u,v_2) \in V_1\times V_2$ for
$u\in \partial^T e_1$ and also to all elements $(e_1,f)\in E_1\times E_2$ for
$f \in \partial v_2$. (Here we view $\partial^T e_1, \partial v_2$ equivalently both as vectors
in $\bbF_2^{V_1}, \bbF_2^{E_2}$ respectively and as subsets of $V_1,E_2$.)  
It follows from this definition that ${\cal C}_{\times}$ is a CSS code
${\cal C}_{\times}(S_x,S_z)$, where as usual $S_x = \Img{H_x}$ and $S_z = \Img{H_z}$.
For
$|V_1| = n_1, |V_2| = n_2, |E_1| = m_1, |E_2| = m_2$, the code ${\cal C}_{\times}$ is a
quantum CSS code on $n_1 n_2 + m_1 m_2$ qubits, with $n_1 m_2 + n_2 m_1$ local checks.
One can check that $\calC_\times$ is determined only by $\calC_1,\calC_2$ and not the
specific choices of $\partial_1,\partial_2$, so \eq{TZ-notation} is well defined.

\begin{figure}[h]\label{fig:tzconst}
\center{
 \epsfxsize=2in
 \epsfbox{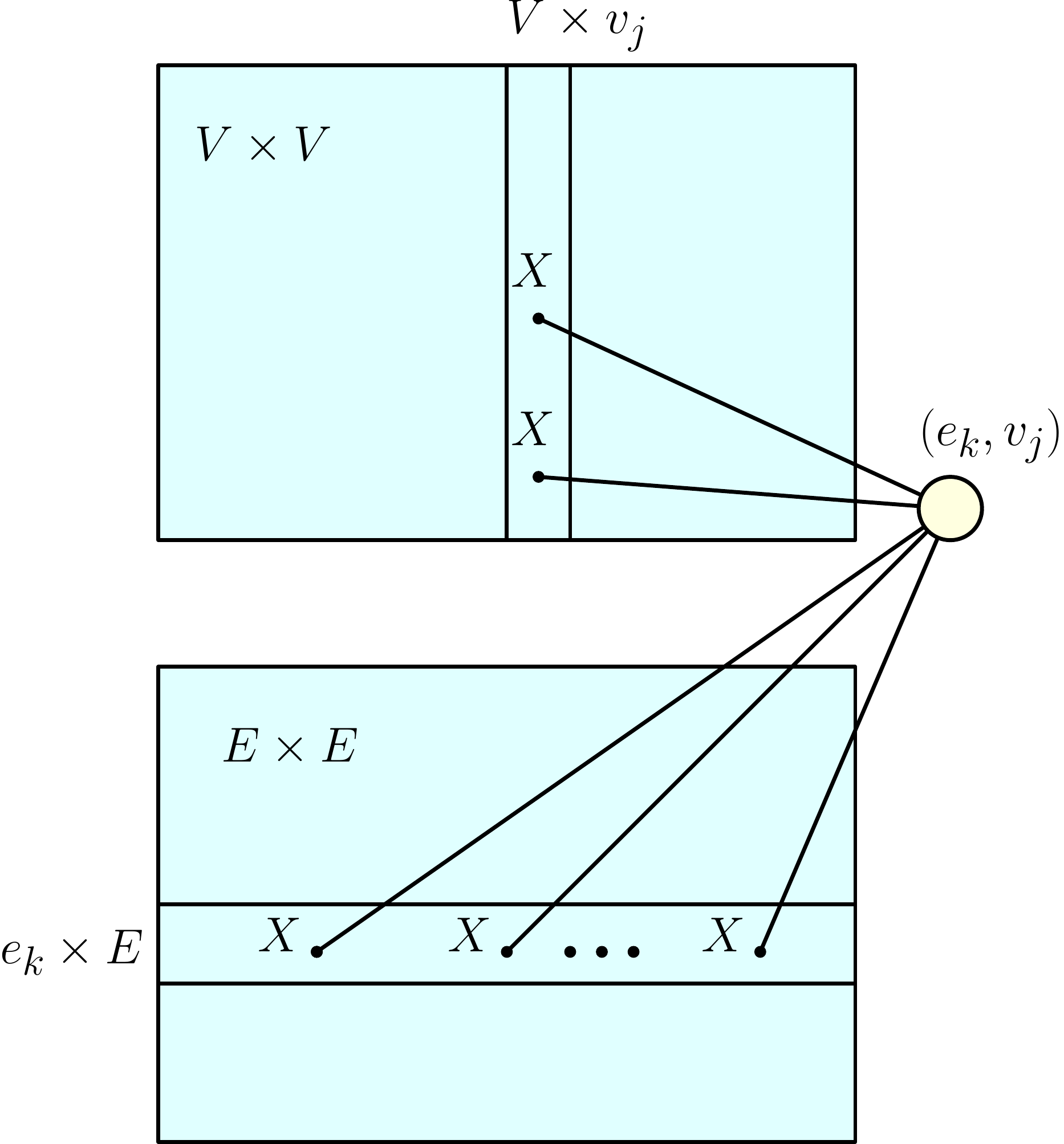}
 \caption{
\small{An example of a check term $(e_k,v_j)$ of $H_x$.  It is a parity check on
     all bits $(v_m,v_j)$ in the $j$-th column of $V\times V$ such that $v_m$ is examined
     by $e_k$ in the original code ${\cal C}$, 
 and on all bits in the $k$-th row of $E\times E$ that corresponds to checks incident on
 $v_j$ in ${\cal C}$.
If we specialize to the case when $\calC$ is the repetition code with checks corresponding
to a $d$-local graph (as in \secref{TZ-expander}) then each check
examines two bits in the
$V\times V$ block and $d$ bits in the $E\times E$ block.
}}} 
\end{figure}

We now state several useful facts on this construction, which can all be found in \cite{TZ09}:
\begin{fact}\label{fact:TZ}

\textbf{Basic Properties of the hypergraph product}\cite{TZ09}

\noindent
\begin{enumerate}
\item
If $C_1,C_2$ have locality parameters $l_1,l_2,l_1^T,l_2^T$, respectively, 
($l_i^T$ is the maximum number of checks incident upon any bit in code $C_i$)
then
$\calC_{\times}$ has locality parameter $l_1 + l_2^T$ for $H_x$, and $l_2 + l_1^T$ for $H_z$.
\item
$
\delta_{\min}({\cal C}_{\times})
\geq
\min \left\{ \delta_{\min}({\cal C}_1),\delta_{\min}({\cal C}_2),\delta_{\min}({\cal C}_1^T),\delta_{\min}({\cal C}_2^T)\right\}
$
\item Let $r(\calC)$ denote the number of qubits in a code $\calC$.
Then
$
r({\cal C}_{\times}) = r({\cal C}_1) \cdot r({\cal C}_2) + r({\cal C}_1^T) \cdot r({\cal C}_2^T).
$
\end{enumerate}
\end{fact}
These logical operators of ${\cal C}_{\times}$ can assume very complex forms, due in part, to the fact that the rate of the code
scales like $r({\cal C}_1) \cdot r({\cal C}_2)$. Hence, the hypergraph product of codes with linear rate is linear itself,
i.e.~scales like $\Omega(|V|^2)$.

\subsection{Column-wise logical operators}\label{sec:column-logical}

A particularly interesting subset of the logical operators, which is a subgroup w.r.t.~addition modulo $\bbF_2$,
has a very succinct and useful form.
We exploit the structure of this group to inherit,  in some sense,
the classical property of local testability.  
\begin{fact}\label{fact:logicalshape}

\textbf{Group of logical operators isomorphic to the original code}

\noindent
For any $x\in {\cal C}_1$, and $y\notin {\cal C}_2^{\perp}$, the word 
\be
\left(\left(x \ot y\right)_{V_1\times V_2} , \mathbf{0}_{E_1\times E_2}\right)
 \in S_x^\perp - S_z
\ee
Similarly, for $x\notin \calC_1^\perp, y\in\calC_2$, 
\be\left(\left(x \ot y\right)_{V_1\times V_2} , \mathbf{0}_{E_1\times E_2}\right)
 \in S_z^\perp - S_x.
\ee
One can also show that
\begsub{cycles}
 (\mathbf{0}_{V_1\times V_2}, \calC_2^T \ot (\calC_1^{T\perp})^c) & \subset
 S_x^\perp - S_z \\
 (\mathbf{0}_{V_1\times V_2}, (\calC_2^{T\perp})^c \ot \calC_1^T) & \subset
 S_z^\perp - S_x 
\endsub
In particular, if ${\cal C}_1,{\cal C}_2,{\cal C}_1^T,{\cal C}_2^T$ 
are linear codes in which each bit appears at least once as $0$ and once as $1$ in some non-zero word,
then 
\begsub{special}
\left ((\calC_1 \ot \F_2^{V_2})_{V_1\times V_2}, \mathbf{0}_{E_1\times E_2}\right)
& \subset  S_x^\perp - S_z \\
\left ((\F_2^{V_1} \ot \calC_2)_{V_1\times V_2}, \mathbf{0}_{E_1\times E_2}\right)
& \subset  S_z^\perp - S_x
\endsub
\end{fact}
\noindent
The proof of this fact is straightforward and can be found in
\cite{TZ09}.

\subsection{The Hypergraph Product of a Connected Graph}
\label{sec:TZ-expander}

\begin{proposition}\label{prop:prodexp}

\textbf{The hypergraph product of a connected graph}

\noindent
Let $G = (V,E)$ denote a $d$-regular connected graph on $n$ vertices.
Let ${\cal C} = {\cal C}(G)$ denote the repetition code on $n$ bits defined by treating the edges of $G$ as equality constraints.
Let ${\cal C}_{\times}(G)$ denote the hypergraph product of ${\cal C} \times_{TZ} {\cal C}$.
Then:
\begin{enumerate}
\item
Denote $|V| = n$, $|E| = m = dn/2$, and so
$|V\times V| = n^2, |E\times E| = d^2n^2/4, |V\times E| = |E\times V| = dn^2/2$.
The number of qubits is $N = (1+d^2/4)n^2$ and the number of checks is $dn^2$.
\item
${\cal C}_{\times}$ is a quantum code 
on the space of $\bbF_2^{V\times V}\oplus\bbF_2^{ E\times E} = \F_2^{N}$,
constrained by the $d+2$-local checks from the columns of $\{H_x,H_z\}$.
\item
The following set of vectors, indexed by $v\in V,e\in E$,  generates $S_z$, 
\be\label{eq:sz}
s_z(v,e) =  H_z^T(v\ot e) = v \otimes \partial^Te + \partial v \otimes e
\ee
Likewise $S_x$ is generated by the vectors
\be \label{eq:sx}
s_x(e,v) = H_x^T(e\ot v) = \partial^Te\otimes v + e\ot\partial v. \ee
\item
$\dim(S_x^{\perp} / S_z) = 1 + \dim({\cal C}^T)^2$.
This follows from Proposition 14 in \cite{TZ09}.  

\item
The distance of the code is 
given by the minimum of the distance of the code ${\cal C}$ and the transposed code ${\cal C}^T$.
\end{enumerate}
\end{proposition}

We can also specialize our characterization of logical operators from \factref{logicalshape}
to the repetition code with 2-bit check operators.
\begin{proposition}\label{prop:tz1}
Let ${\cal C}_\times = {\cal C}\times {\cal C} = {\cal C}_{\times}(G)$, where $G$ is a connected graph,
and ${\cal C}(G)$ is the repetition code constrained by parity checks corresponding to the edges of $G$.
There exists a spanning set ${\cal B}_z$ of $S_x^{\perp}$, and a spanning set $\calB_x$ of
$S_z^\perp$, as follows:
\begsub{spanning}
{\cal B}_z &:= \{ b_1^z\} \cup \{ e \ot c\}_{e\in E,c\in {\cal C}^T} \cup S_z \\
{\cal B}_x &:= \{ b_1^x\} \cup \{ c \ot e\}_{e\in E,c\in {\cal C}^T} \cup S_x 
\endsub
where 
${\cal C}^T = \ker \partial^T$ denotes the linear span of all indicator vectors of edges
corresponding to cycles in $G$.
\end{proposition}

\begin{proof}
  This follows from \cite{TZ09} as follows.  From the proof of Lemma 17 of \cite{TZ09} we
  have that $b_1^z$ and $e\ot c$, for each $c\in {\cal C}^T$ are in $S_x^{\perp} - S_z$.
  By Proposition 14 of \cite{TZ09} it follows that these words, with $S_z$, span the
  entire $S_x^{\perp}$ space.  The argument for $\calB_x$ is the same.
\end{proof}

\subsubsection{Fractal Structure}

Another important property of the hypergraph product of a connected graph, is that  the
hypergraph product exhibits a fractal structure as follows:

\begin{proposition}\label{prop:fractal}
  Let $G = (V,E)$ be some graph, and let ${\cal C}_\times(G)$ denote the hypergraph product of the repetition code induced by equality constraints of $E$, with itself.  Let $V_l\subseteq V, E_l\subseteq E$ denote some subsets.  Then there exists a graph $G' = (V_l, E_l \cap V_l\times V_l)$ such that ${\cal C}_{\times}(G')$ is supported on $V_l \times V_l \cup E_l\times E_l$.
\end{proposition}

\begin{proof}
By definition, the checks of ${\cal C}_{\times}$ are the Cartesian product $E\times V$
for $S_x$ and $V\times E$ for $S_z$.
Define $G' = (V',E')$ as in the statement of the proposition, i.e.~with 
$E'$ the set of edges in $E_l$ that have both endpoints in $V_l$.
Hence $E'\subseteq E$, $V'\subseteq V$, and so in particular
$E' \times V' \cup V'\times E' \subseteq E\times V\cup V\times E$.
\end{proof}

\subsection{The Hypergraph Product of an Expander Graph}

In this section, we consider the hypergraph product ${\cal C}_\times(G) = {\cal C}(G) \times {\cal C}(G)$, where $G$ is a $d$-regular
Ramanujan expander graph.
We note that while the minimal distance of ${\cal C}$ is exactly $n$, as it is the repetition code, the 
minimal distance of ${\cal C}^T$ is much smaller, i.e.~$O(\log(n))$ - given by the minimum length
cycle in the expander graph. Hence
$$
\delta_{\min}(\calC_\times) = \min\{\delta_{\min}({\cal C}), \delta_{\min}({\cal C}^T)\} = \min\{O(n), O(\log(n))\} = O(\log(n)).
$$

\subsubsection{Comparison to the Toric Code}
One can first compare ${\cal C}_\times(G)$ to the Toric Code.
The Toric Code can be seen as the hypergraph product of the repetition code, with equality
constraints in a cycle, i.e.~$x_1=x_2, x_2=x_3,\ldots, x_n=x_1$.  (By contrast our code
has equality constraints $x_i=x_j$ for $(i,j)$ running over the set of edges in an
expander graph.)
It follows from the hypergraph product, that the distance of such a code is precisely $n$
(out of $n^2$ total qubits), which is larger than the $O(\log n)$ minimum distance of our
code. However, the toric code also has low-error trivial states, since we can delete an
$O(\eps)$ fraction of constraints and leave it disconnected into blocks of $1/\eps^2$
qubits.

\subsubsection{Localized Minimal Distance}

As stated above we have $\delta_{min}({\cal C}_\times) = O(\log(n))$.
However,
not all logical qubits are equally protected: we
focus on the logical qubit with distance $n$, meaning that
all elements of $C_1^Z$ and $C_1^X$ have weight $\geq n$.  
It turns out, that for this logical qubit, an even stronger property is true:
we will show that any element of $C_1^Z$ or $C_1^X$
must have weight $\Omega(n)$ in some row or column of $V\times V$ or $E\times E$.  
In other words, the minimal distance of this logical qubit is manifested {\it locally}:

\begin{lemma}\label{lem:localized-dist}

\textbf{Locally-manifested minimal distance}

\noindent
Let ${\cal C}_{\times}(V'\times V' \cup E'\times E') = {\cal C}(G')$ denote the hypergraph product of  
a graph $G' = (V',E')$,
which is a connected $\eps$-residual graph of a Ramanujan graph of degree $d$.
If $d\geq 14$ and  $\eps\leq \frac{1}{6200d}$ then
\be 
\forall w\in C_1^Z \quad
\left(\exists v\in V' \quad
|w_{V'\times v}| \geq \frac 12 n'
\qquad\text{or}\qquad
\exists e\in E'\;\;|w_{E'\times e}| \geq \frac{3}{8d} n'\right).\ee
where $n'=|V'|$.  Similarly,
\be 
\forall w\in C_1^X \quad
\left(\exists v\in V' \quad
|w_{v\times V'}| \geq \frac 12 n'
\qquad\text{or}\qquad
\exists e\in E'\;\;|w_{e\times E'}| \geq \frac{3}{8d} n'\right).
\ee
\end{lemma}

\begin{proof}
Consider $w\in C_1^Z$.  The proof for $C_1^X$ is essentially the same so we omit it.
Since ${\cal C}_{\times} = {\cal C}_\times(G')$ for a connected graph $G'= (V',E')$ then by
\propref{tz1}, we can represent $w$ as the following sum:
\be w = b_1^z + s +c ,\quad b_1^z = \mathbf{1}_{V'\times v_1}, \quad s \in S_z, \quad c \in (\mathbf{0}_{V'\times V'}, \F_2^{E'} \ot \ker \partial^T),
\ee
and we assume w.l.o.g. that $c$ has minimal weight modulo $S_z$.
By the hypergraph product we can write $s$ as
\be\label{eq:s-alpha}
s 
= \sum_{v\in V', e\in E'} \alpha_{ve} s_z(v,e)
 = (I_{V'} \ot \partial^T + \partial \ot I_E') \alpha,
\ee
for some $\alpha\in\F_2^{V'\times E'}$. 
Let us focus initially on the $V'\times V'$ block:
\be w_{V'\times V'} = \mathbf{1}_{V'\times v_1} + (I \ot \partial^T)\alpha.\ee
Hence,
\be 
\forall v\in V' \quad
w_{V'\times v} = \delta_{v,v_1} \mathbf{1}_{V'} + \bigoplus_{e\in\partial v} \alpha_{V'\times e}
\label{eq:w-column}.
\ee
Suppose that for all $v\in V'$, this column has low weight, i.e.~
\be\label{eq:lowv}
 |w_{V'\times v}| <   \frac 12 n'\label{eq:v-col-light},
 \ee 
(Otherwise we are done.)  
For each edge $e$, we consider
$|\alpha_{V'\times e}|$ and define as follows:  
\bas
e \text{ is light} & && \frac{|\alpha_{V'\times e}|}{n'} < \frac{1}{2d} \\
e \text{ is medium} && \frac{1}{2d} \leq &\frac{|\alpha_{V'\times e}|}{n'}
\leq 1-\frac{1}{2d} \\
e \text{ is heavy} && 1-\frac{1}{2d}< &\frac{|\alpha_{V'\times e}|}{n'}
\eas
We claim that there exists at least one medium edge.  Suppose, towards a contradiction,
that all edges are light or heavy.
For each $v$, let $\xi(v)$ denote the number of heavy edges incident upon $v$.  

We now focus on the second term in the RHS of \eq{w-column}.  From the triangle inequality,
\be   \left\{
\begin{array}{ll}
      \left|\sum_{e\in\partial v} \alpha_{V'\times e}\right| > \frac {n'}{2} & \text{ if }\xi(v) \mbox{ odd } \\
      \left|\sum_{e\in\partial v} \alpha_{V'\times e}\right| < \frac {n'}{2} & \text{ if }\xi(v) \mbox{ even }
\end{array} 
\right. \ee

Since by \eq{v-col-light} the Hamming weight in {\it each} column is $<\frac {n'}{2}$
we conclude
that $\xi(v_1)$ is odd and all other vertices have an even value of $\xi(v)$.  
On the other hand, since
each heavy edge is incident upon two vertices we have that $\sum_{v\in V'} \xi(v)$ is even.
This is a contradiction, and so we conclude that there exists at least one
medium edge.  Let $e_0$ denote a medium edge.
\be \label{eq:e0}
\exists e_0\in E' \quad
\frac{n'}{2d} \leq |\alpha_{V'\times e_0}| \leq
\left(1 - \frac{1}{2d}\right)n'. 
\label{eq:medium}\ee
We turn now to the $E'\times E'$ block.
By the hypergraph product and \eq{s-alpha}:
\be |s_{E'\times e_0}| = |\partial \alpha_{V' \times e_0}|.
\ee
By minimality of $c$ modulo $S_z$ we claim that
\be\label{eq:s+c}
| (s+c)_{E'\times e_0}| \geq \frac 14 |\partial \alpha_{V'\times e_0}|,
\ee
Suppose, towards contradiction, that $c_{E'\times e_0}$ is $|\partial
\alpha_{V'\times e_0}|/4$-close to $\partial \alpha_{V'\times e_0}$, 
and let
\[   \tilde\alpha := \left\{
\begin{array}{ll}
      \alpha_{V'\times e_0} &  |\alpha_{V'\times e_0}| \leq n'/2\\
       \mathbf{1}_{V'} + \alpha_{V'\times e_0} & \mbox{ o/w }.
\end{array} 
\right. \]
By \eq{e0} we have $|\tilde\alpha| \in [n'/2d, n'/2]$.
By Corollary \ref{cor:cheeger1} 
for $d=14$, and $\eps \leq 1/6200$ we have:
\be\label{eq:rho2}
\forall x\in \F_2^{V'}, \ 3100\eps n' \leq |x| \leq n' / 2,\qquad  |\partial x | \geq  3 \cdot x,
\ee
Choosing $\eps \leq \frac{1}{6200d}$ (here $\eps=10^{-9},d=14$ works) implies that Corollary \ref{cor:cheeger1} applies to $\tilde\alpha$: 
\be\label{eq:ltc3}
|\partial \tilde\alpha | \geq
3 \cdot |\tilde\alpha|
\ee
Consider now the following stabilizer word:
$$
s = \sum_{v, \tilde\alpha(v) = 1} s_z(v,e_0) \in S_z
$$  
We claim it reduces the weight of $c$, by
contradiction to $c$'s minimality modulo $S_z$, as follows:
By the triangle inequality:
\be
\left|c + \sum_{v, \tilde\alpha(v) = 1}  s_z(v,e_0)\right|
\leq
|c| + \frac{1}{4}| \partial \tilde\alpha| - | \partial \tilde\alpha| + 2 |\tilde\alpha| 
\ee
where the first three terms come from our assumption that $c_{E'\times e_0}$ is 
$|\partial \tilde\alpha|/4$-close to $\partial \tilde\alpha$, and the last term from the fact that
$s_z(v,e_0)_{V'\times V'}$ has weight exactly $2$,
By \eq{ltc3} we upper-bound the above by:
\be
|c| -  \frac{1}{3} |\partial \tilde\alpha| +  \frac{1}{4} |\partial \tilde\alpha| 
<
|c|.
\ee
Hence \eq{s+c} holds, and so:
\be\label{eq:lowe}
| w_{E'\times e_0}|
=
| (s+c)_{E'\times e_0} |  \geq 
\frac 14 |\partial \alpha_{V'\times e_0}| \geq
\frac{3}{4} |\tilde\alpha|
{\geq} 
\frac{3}{8d} n'. \ee

\end{proof}

\section{Explicit $\QNC^1$-Robust Local Hamiltonians}
\label{sec:nlets}
\cftsectionprecistoc{We describe our construction of Hamiltonians with
  $\cNLTS$.}

\subsection{The construction}

In this section, we show how to construct $\QNC^1$-robust local
Hamiltonians based on CSS codes.
Let ${\cal G}$ be an explicit family of $d$-regular Ramanujan graphs, for $d = q+1$, where $q$ is prime,
following Fact \ref{fact:explicit}.
We define
\be \calC_{\times} = \calC(G).
\ee

\subsection{$\NLETS$ Theorem Statement}

\begin{theorem}

\textbf{$\NLETS$}

\noindent
Let ${\cal C}_\times^{(N)}$ denote the hypergraph product above that is defined on a space
of $N = (1+d^2/4)n^2$ qubits.
The family of
local Hamiltonian $\left\{H \left({\cal C}_{\times}^{(N)}\right)\right\}_N$ is $\NLETS$
for $d=14$ and $\eps= 10^{-9}$.
\end{theorem}
(Our proof applies to any $d\geq 14$ and sufficiently small $\eps>0$,
which may depend on $d$.)

\noindent
The proof has $3$ steps.
\begin{enumerate}
\item In the first part of the proof we show that any quantum state $\ket{\psi}$
  that is an $\eps$-impostor of $H({\cal C}_\times)$ obeys, in fact, a more
  stringent constraint on a {\it subsystem} of the full Hilbert space, 
  which is the uniform low-weight error condition: there exists some large
  subset $V_l\subseteq V$ and large subset $E_l\subseteq E$ such that for {\em each} $v\in V_l$
  at most an $O(\sqrt{\eps})$ fraction of the qubits $V\times v$ have errors,
  and the same holds for a large fraction of columns $E\times e$, for $e\in E_l$.
  
  A certain ``fractal'' property of the hypergraph product of the repetition code
  defined by an expander graph,
  allows us to argue that inside ${\cal C}_{\times}$
  there exists a complete smaller product-hypergraph 
  of a connected sub-graph $(V',E')$ induced by $(V_l, E_l)$.
  Hence, we can reduce the problem of an $\eps$-impostor to the hypergraph product code of $G$
  to the problem of an $\eps$-impostor to the hypergraph product code of $G_l$, with the extra
  condition of uniform low-weight error.
  
\item In the second part we show that this ``uniform low-weight error condition'' implies
that there exists a distance partition in either the $X$ or $Z$ basis.  
\item In the third part, we finish the proof by using the distance partition above to
  argue that the $\ket{\psi}$ has low vertex expansion in at least the $X$ or $Z$ basis.
\end{enumerate}

\subsection*{Part 1: Uniform Low-weight Error on a Sub-code}

Fix some integer $N>0$.  We think of $N$ as being sufficiently large given any choice of the other parameters.
We denote ${\cal C}_\times = {\cal C}_\times^{(N)}$ for simplicity.
We are given a quantum state $\rho$ that is an $\eps$-impostor to $H({\cal C}_\times)$.
To establish the theorem we show that $\rho$ is $\QNC_1$-hard.

The fact that $\rho$ is an $\eps$-impostor means that there exists a state $\sigma$ with $\tr[H\sigma]=0$ and a set $S\subseteq (V\times V) \cup (E \times E)$ such that $|S| \geq N(1-\eps)$ and $\rho_S=\sigma_S$.  Here we use the convention that $\rho_S := \tr_{\bar S}\rho$.

Now define sets $V_l \subseteq V, E_l \subseteq E$ by
\ba
V_l &= \left \{ v :
|V\times v \cap S| \geq (1-d\sqrt{\eps}) n
\quad\wedge\quad
|v\times V \cap S| \geq (1-d\sqrt{\eps}) n
 \right\}\\
E_l &= \left \{ e :
|E\times e \cap S| \geq (1-d\sqrt{\eps}) m 
\quad\wedge\quad
|e\times E \cap S| \geq (1-d\sqrt{\eps}) m 
\right\}
\ea
By Markov's inequality we can bound $|V_l| \geq (1-d\sqrt\eps)n$ and
$|E_l| \geq (1-d\sqrt\eps)m$.

Let $G' = (V',E')$ denote the 
{\it maximal connected} residual graph of $G$ 
(see Definition \ref{def:res1})
induced by the vertices $V_l$, and edges $E_l$.
By Fact \ref{fact:eml1}, \eq{d=14} and the fact that $d=14$,  $V',E'$ satisfy
$|V'| \geq (1 - 31d\sqrt\eps)n$ and $|E'| \geq (1 - 62d\sqrt\eps) m$.
Let $S' := (V'\times V')\cup(E'\times E')$ be the corresponding qubits.

By Proposition \ref{prop:fractal} the full hypergraph product ${\cal
  C}_{\times}$ contains {\it all}
parity checks of the following hypergraph product:
\be
{\cal C}_{\times}^{\text{(low)}}= {\cal C}(G') \times_{TZ} {\cal C}(G').
\ee 
Now decompose the Hilbert space ${\cal H}$ of ${\cal C}_\times$ according to the uniform low-violation condition:
\be
{\cal H} = {\cal H}_h \otimes {\cal H}_l,
\qquad\text{where }\qquad
{\cal H}_l:= \bbC_2^{S'}
\ee
In words, ${\cal H}_l$ is the support of ${\cal C}_{\times}^{\text{(low)}}$ and contains
only (though not all) qubits that belong to columns/rows of $V\times V$ or $E\times E$ with
individually low-weight error, and ${\cal H}_h$ contains qubits in high-error columns/rows over which we have no control.  

When we restrict our attention to the qubits in $S'$, we find that $\rho_{S'}$ satisfies a stronger version of the impostor 
condition for $H({\cal C}_\times^{\text{(low)}})$.  Define $\nu :=  \frac{d\sqrt{\eps}}{1-62d\sqrt\eps}$.   Since $S\cap S'$ contains a $\geq 1-\nu$ fraction of the bits of each row and column of $S'$, then $\rho_{S'}$ agrees with some ground state of $H({\cal C}_\times^{\text{(low)}})$ (namely $\sigma_{S'}$) for a $\geq 1-\nu$ fraction of  each row and column.

\subsection*{Part 2: Distance Partition on the Subcode}

Let $n' = |V'|$ and  $m' = |E'|$.
Let $S_{x,\text{low}},S_{z,\text{low}}$ denote the set of generators of ${\cal
  C}_{\times}^{\text{(low)}}$,  
and let $b_{1,\text{low}}^x$ be the corresponding element of $S_{x,\text{low}}^{\perp} - S_{z,\text{low}}$,
associated with the first column of the register ${\cal H}_l$:
i.e.~$b_{1,\text{low}} = \mathbf{1}_{V'\times v_1}$ - namely - the word which is $1$ in $V'\times v_1$ and zero otherwise, for  some arbitrary $v_1\in V'$.
Let $N_{\text{low}}$ denote the number of qubits in ${\cal C}_{\times}^{\text{low}}$.
Let $\partial_{G'}: \F_2^{V'} \mapsto \F_2^{E'}$ denote the vertex-edge
incidence operator corresponding to the sub-graph $G'$ defined above.

\begin{definition}

\textbf{Uniform low-weight errors}

\noindent
Define the sets of uniform low-weight errors in the $X$ and $Z$ bases
to be the following subsets of $\F_2^{V'}$ 
\begsub{ULWE}
U^{z,\nu} & := \left\{ t:  |t_{E'\times e}| \leq \nu m' \
  \forall e\in E' \ \  \wedge 
 \ \  |t_{V'\times v}| \leq \nu n' \  \forall v\in V' \right\} \\
U^{x,\nu} & := \left\{ t:   |t_{e\times E'}| \leq \nu m' \
  \forall e\in E' \ \  \wedge 
 \ \  |t_{v\times V'}| \leq \nu n' \  \forall v\in V' \right\} 
\endsub

\end{definition}
We can formalize the uniform low-weight error condition from Part 1 by observing that 
\be x\in \bbF_2^N, \supp (x)\subseteq {\bar S \cap S'}\quad\Rightarrow\quad
x \in U^{z,\nu}\cap U^{x,\nu} .\label{eq:S-ULWE}\ee

\begin{definition}
\textbf{Distance partition}

\noindent
We define $S_0^{z,\nu}, S_1^{z,\nu}$ 
as the strings obtained by adding uniform low weight error to a code
state as follows.
 For $a=0,1$ define
\begsub{s0}
S_a^{z,\nu} & = C_a^Z + U^{z,\nu} \\
S_a^{x,\nu} & = C_a^X + U^{x,\nu} 
\endsub
\end{definition}
Since $C_{0,1}^Z$ partitions $S_{x,\text{low}}^\perp$ (and likewise
for $C_{0,1}^X$) we have
\begsub{part2}
S_{x,\text{low}}^\perp + U^{z,\nu}  &= S_0^{z,\nu} \sqcup S_1^{z,\nu}\\
S_{z,\text{low}}^\perp + U^{x,\nu}  &= S_0^{x,\nu} \sqcup S_1^{x,\nu}
\endsub
Here $C_{0,1}^{X,Z}$ are defined as in \eq{CXZ} but with respect to the code $\calC_\times^{\text{low}}$.

\noindent
We claim that these sets are far apart:
\begin{fact}\label{fact:partition1}

\textbf{Uniform low-weight errors preserve the distance partition.}
\noindent
\be 
\Delta(S_0^{z,\nu},S_1^{z,\nu}) = \Omega(n),
\ee
and the same holds for $S_0^{x,\nu}, S_1^{x,\nu}$.
\end{fact}

\begin{proof}

Let $x_0\in S_0^{z,\nu}, x_1\in S_1^{z,\nu}$.
By the definition of $S_{0,1}^{z,\nu}$ there exist
$\hat{x}_0 \in C_0^Z,\hat{x}_1\in C_1^Z$ such that $x_0 - \hat{x}_0,
x_1 - \hat{x}_1 \in U^{z,\nu}$.

We note that $\hat{x}_0 \in C_0^Z,\hat{x}_1\in C_1^Z$ where ${\cal C}(G')$ is the hypergraph product
of $G'$ which is a (maximal) connected residual graph (see Definition \ref{def:res1}) of a $d$-regular Ramanujan graph.
Thus,
we can invoke
Lemma~\ref{lem:localized-dist} 
which implies that either there exists $v\in V'$
such that $|(\hat{x}_0)_{V'\times v} + (\hat{x}_1)_{V'\times v}| \geq \frac {n'}{2}$
or there exists some $e\in E'$ such that
 $| (\hat{x}_0)_{E'\times e} + (\hat{x}_1)_{E'\times e}|\geq \frac{3}{8d}  n'$.
The triangle inequality then implies
\be
|x_0 + x_1| \geq \min\left\{ \frac{1}{2} n' - \nu n', \frac{3}{8d}n' - \nu m'\right\}= \Omega(n),
\ee
for any constant $\nu < \frac{3/8d}{m'/n'}$.  This last bound is satisfied if we take $\eps \leq 10^{-9}$.

\end{proof}

\item
\subsection*{Part 3: Quantum Circuit Lower-bound}

Let $p_Z$ denote the distribution induced by measuring all $N_{\text{low}}$ qubits of $\rho_{S'}$ in the $Z$ basis.
If we had measured $\sigma_{S'}$ instead then the measurement outcomes would have been entirely in $C_0^Z\cup C_1^Z$.  From Part 1, we have that the outcomes from measuring $\rho_{S'}$ are the same on a $\geq 1-\nu$ fraction of the bits of each row and column.  Therefore these outcomes differ  by an element of $U^{z,\nu}$ and we have:
\be\label{eq:pzi}
\supp(p_Z) \subseteq S_0^{z,\nu} \cup S_1^{z,\nu}.
\ee
(We can define $p_X$ similarly and likewise observe that $\supp(p_X) \subseteq S_0^{x,\nu} \cup S_1^{x,\nu}.$).
Our goal is now to show that either $p_Z$ or $p_X$ is $\QNC^1$-hard.

By Fact \ref{fact:partition1} we have:
\be\label{eq:s0s1}
\Delta(S_0^{z,\nu},S_1^{z,\nu}) = \Omega(n)
\ee
In particular $S_0^{z,\nu} \cap S_1^{z,\nu} = \emptyset$,
where $S_0^{z,\nu} = C_0^Z + U^{z,\nu}, S_1^{z,\nu} = C_1^Z + U^{z,\nu}$,
and the same holds for $S_{0,1}^{x,\nu}$.
Hence, we can invoke
\propref{vor1} 
choosing the error sets as $E_z = U^{z,\nu}, E_x = U^{x,\nu}$.
This implies we have uncertainty in either the $X$ or $Z$ bases.  Without loss of generality, assume that we have uncertainty in the $Z$ basis, i.e.:
\be\label{eq:p3}
p_Z(S_0^{z,\nu}) \geq c_0, \ \ 
p_Z(S_1^{z,\nu}) \geq c_0.
\ee

Now let $\tilde{p}$ be a distribution that approximates $p_Z$:
$$
\| \tilde{p} - p_Z \|_1 \leq N_{\text{low}}^{-a}.
$$
so that 
\ba\label{eq:close2}
\tilde{p}(S_0^{z,\nu} \cup S_1^{z,\nu}) &\geq 1 - N_{\text{low}}^{-a} \\
\label{eq:p3}
\tilde{p}(S_0^{z,\nu}) &\geq c_0 - N_{\text{low}}^{-a} \geq c_0/2, \\ 
\tilde{p}(S_1^{z,\nu}) &\geq c_0 - N_{\text{low}}^{-a} \geq c_0/2.
\ea

Eqs.~\eq{close2}, \eq{p3} and \eq{s0s1} imply together in particular
that $S_0^{z,\nu}$ has a very low vertex expansion for some $l = \Omega(n)$:
\be
\frac{\tilde{p}(\partial_l(S_0^{z,\nu}))}{\tilde{p}(S_0^{z,\nu})} \leq N_{\text{low}}^{-a},
\ee
Hence, the vertex expansion of the distribution $\tilde p$ is upper-bounded by: 
\be
h_{\ell}(\tilde{p}) \leq N_{\text{low}}^{-a} ,\quad\text{for}\, \ell = \Omega(n).
\ee
We now invoke Theorem~\ref{thm:vertex}.
By the theorem for any $\alpha\in [0,1/2)$ there is some $\ell>0$ at most
\be\label{eq:ell1}
\ell \leq B \cdot (B N_{\text{low}})^{1/2 - \alpha/2},
\ee
for which $h_{\ell}(\tilde{p}) \geq \frac{1}{2}(B N_{\text{low}})^{-\alpha}$.
Since $a>0$ is a constant independent of $n$ then we can choose a constant $\alpha>0$ so that
\be
n^{-a/2} \leq \frac{1}{2} (B N_{\text{low}})^{-\alpha} \leq h_{\ell}(\tilde{p}).
\ee
Hence $\ell$ is upper-bounded by \eq{ell1} but must surpass the distance partition
of \eq{s0s1}, i.e.:
\be
\Omega(n) \leq \ell \leq B \cdot (B N_{\text{low}})^{1/2 - \alpha/2}
 = B^{1.5-\alpha/2}\cdot O(n^{1-\alpha})
\ee
Therefore by \eq{ell1}, $B^{1.5 - \alpha/2} \geq n^{\alpha}$.
Hence any quantum circuit $U$ of depth $d$ for which $\ket{\psi} = U \ket{0^{\otimes N}}$, has
$d = \Omega(\log(n)) = \Omega(\log(N))$.
Thus $\ket\psi$ is $\QNC_1$-hard.

\end{document}